\documentclass[11pt]{article}%
\usepackage{eurosym}
\usepackage[T1]{fontenc}
\usepackage[utf8]{inputenc}
\usepackage[top=1in, bottom=1in, left=1in, right=1in]{geometry}

\usepackage[longnamesfirst]{natbib}
\usepackage{comment}
\usepackage{dsfont}
\usepackage{xcolor}

\usepackage{tikz}
\usetikzlibrary{arrows.meta, calc, positioning, patterns,matrix,decorations.pathreplacing}
\usepackage{pgfplots}
\pgfplotsset{compat=1.18}
\usepgfplotslibrary{fillbetween}

\usepackage{graphics}
\usepackage{graphicx}
\usepackage{color}
\usepackage{amsmath}
\usepackage{amsthm}
\usepackage{amssymb}
\usepackage{amsfonts}
\usepackage[normalem]{ulem}
\usepackage{float}
\usepackage[noamsmath]{kpfonts}

\usepackage[font=normalsize,justification=centering]{caption}
\DeclareCaptionLabelFormat{captionf}{\textsc{#1~#2}}
\numberwithin{equation}{section}
\newtheorem{theorem}{Theorem}[section]
\newtheorem{lemma}{Lemma}[section]
\newtheorem{proposition}{Proposition}[section]

\newtheorem{corollary}{Corollary}[section]

\theoremstyle{definition}
\newtheorem{definition}{Definition}[section]

\theoremstyle{remark}

\definecolor{ForestGreen}{rgb}{.13,.54,.13}
\definecolor{BrickRed}{rgb}{.80,.26,.33}
\definecolor{light_gray}{RGB}{230,230,230}
\definecolor{Gray}{RGB}{160,160,160}
\definecolor{Purple}{HTML}{808080}
\definecolor{Black}{HTML}{000000}
\definecolor{Blackish}{HTML}{283747}
\definecolor{Whiteish}{HTML}{979A9A}

\definecolor{myred}{HTML}{E24A33}  % adjust as desired
\definecolor{myblue}{HTML}{348ABD} % adjust as desired

\definecolor{Mahogany}{rgb}{0.55, 0.18, 0.0}
\usepackage{hyperref}
\hypersetup{colorlinks=true,linkcolor=Mahogany,citecolor=Mahogany,urlcolor=Blackish}

\ifdefined\DRAFT
\newcommand{\fed}[1]{{\color{BrickRed}{(\textbf{Fedor:} #1)}}}
\newcommand{\ben}[1]{{\color{blue}{(\textbf{Ben:} #1)}}}
\else
\newcommand{\fed}[1]{}
\newcommand{\ben}[1]{}

\fi

\newcommand{\R}{{\mathbb{R}}}

\newcommand{\E}{{\mathbb{E}}}

\newcommand{\CCE}{\mathrm{CCE}}
\newcommand{\CE}{\mathrm{CE}}
\newcommand{\NE}{\mathrm{NE}}
\newcommand{\IRCP}{\mathrm{IRCP}}

\newcommand{\conv}{\mathrm{conv}}
\newcommand{\supp}{\mathrm{supp}}

\title{Delegation in Strategic Environments and Equilibrium Uniqueness\thanks{We thank Yu Awaya, Zoe Hitzig, Marcin P\k{e}ski, Doron Ravid, Philipp Strack, Juuso Toikka, Tristan Tomala, and Leeat Yariv for helpful comments and suggestions. We also appreciate the input of %seminar and conference audiences 
the audience at ASSA~2026.}}

\author{Fedor Sandomirskiy\thanks{Princeton University. Email: fsandomi@princeton.edu} \and
        Ben Wincelberg\thanks{Caltech. Email: bwincelb@caltech.edu}}

\begin{document}

\maketitle

\begin{abstract}
We ask when a normal-form game yields a single equilibrium prediction, even if players can coordinate by delegating play to an intermediary such as a platform or a cartel. Delegation outcomes are modeled via coarse correlated equilibria (CCE) when the intermediary cannot punish deviators, and via the set of individually rational correlated profiles (IRCP) when it can. We characterize games in which the IRCP or the CCE is unique, uncovering a structural link between these solution concepts. Our analysis also provides new conditions for the uniqueness of classical correlated and Nash equilibria that do not rely on the existence of dominant strategies. The resulting equilibria are robust to players' information about the environment, payoff perturbations, pre-play communication, equilibrium selection, and learning dynamics. We apply these results to collusion-proof mechanism design.
\end{abstract}

\section{Introduction}

Markets can fail or, more generally, lead to outcomes the regulator did not intend when strategic environments leave too much room for coordination. Sellers form cartels, auction bidders form rings, and platforms can steer market participants to supracompetitive pricing and surplus-extraction outcomes. These coordination channels matter because they can dramatically expand the set of possible outcomes relative to the decentralized Nash benchmark. The familiar problem of multiple Nash equilibria amplifies once we admit intermediaries that can correlate play and discipline deviators. New delegation technologies raise the stakes: delegating economic decisions to AI agents and platforms makes it feasible to execute contingent plans at scale, and the same infrastructure that improves matching and pricing can also facilitate tacit coordination.

For analysts and policymakers, the puzzle is therefore not whether coordination can occur, but when it ceases to matter for prediction and design. If a mechanism, market rule, or regulatory environment admits a wide range of coordinated outcomes, then evaluation and implementation become hostage to assumptions about communication protocols, enforcement possibilities, and equilibrium selection. Conversely, if the underlying strategic interaction admits a single outcome even under rich coordination, then the designer can rely on a sharp prediction robust to institutions the designer may not observe or control.

In this paper, we propose an approach to this puzzle that is agnostic about the intermediary’s objective. Rather than modeling a particular cartel agreement, platform algorithm, or communication protocol, we ask a structural question about the strategic environment itself:
\begin{quote}
\emph{When does a strategic interaction deliver a single equilibrium prediction even if players can delegate play to an intermediary with substantial coordinating power?}
\end{quote}
We consider general normal-form games and formalize delegation as an intermediary who is authorized by players to act on their behalf and thus selects a joint distribution over action profiles. The set of outcomes such an intermediary can sustain depends on its enforcement technology, that is, its ability to ensure that players continue to delegate their decisions. At one extreme, the intermediary cannot orchestrate punishment and can only withdraw correlation; delegation then corresponds to the set of coarse correlated equilibria (CCE) \citep*{hannan1957approximation,moulin1978strategically}. At the other extreme, the intermediary can coordinate all non-deviators on an adversarial response, in the spirit of the individually rational constraint in the folk theorem. Delegation then corresponds to the set of individually rational correlated profiles (IRCP) \citep*[][Chapter 6]{myerson1991book}, that is, correlated outcomes that give each player at least what she can guarantee against hostile play. In contrast with the common perspective treating these sets as permissive benchmarks, we show that they can collapse to a point in economically relevant environments, and we characterize exactly when this occurs.

Our first group of results studies the strongest intermediaries. We show that if IRCP is a singleton, then the implied outcome must be a pure Nash equilibrium. We then provide a complete characterization (Theorem~\ref{th:IRCP}): a game has a unique IRCP if and only if, after a positive affine transformation of utilities, it becomes an \emph{enforcement game}. 

In an enforcement game there is a designated \emph{self-enforcing} action profile $a^*$ with three properties: (i) each player's utility at $a^*$ is normalized to a common reference level; (ii) each player $i$ can guarantee at least that level unilaterally by playing $a_i^*$; and (iii) $a^*$ is the unique maximizer of utilitarian welfare. To see why $a^*$ must be a unique IRCP, note that any distribution that places positive weight away from $a^*$ lowers expected welfare; hence some player must fall below the utility that she can guarantee by playing $a_i^*$ and would reject delegation. Thus in such environments even the most powerful intermediary cannot sustain any outcome other than $a^*$. The other direction---that such games exhaust those with unique IRCP---is the essence of the theorem.

The nature of enforcement games is distinct from the dominance condition commonly used to ensure equilibrium uniqueness. A dominant action profile trivially yields unique Nash, correlated, and coarse correlated equilibria, but it need not yield a unique IRCP: the prisoner's dilemma is the canonical counterexample, where a Pareto improvement over the dominant action profile remains individually rational and thus survives as a delegated outcome. Enforcement games rule out such improvements not by strengthening incentives pointwise but by creating a negative-sum externality: once play leaves the self-enforcing profile $a^*$, expected welfare falls below what compliance guarantees each player. From a design perspective, the characterization tells a planner how to achieve enforcement by targeted payoff modifications (e.g., transfers or fines) that are often far smaller than those needed to create dominance.

The second set of results studies weak intermediaries that cannot coordinate punishment. 
We show that CCE uniqueness has an equally sharp structure. When the unique CCE is pure, it arises exactly in games that are strategically equivalent to enforcement games (Theorem~\ref{th:CCE}). Strategic equivalence allows player-specific payoff shifts that depend only on others’ actions. Such shifts preserve action-by-action comparisons relevant for CCE but affect utility levels induced by hostile play of others and thus are not respected by IRCP. This result yields an unexpected structural link between our two extremes: a unique CCE is exactly the unique IRCP of a suitably transformed game. This link provides an easy-to-check criterion for uniqueness of CCE and, from a design perspective, shows how to modify payoffs to make equilibrium outcomes additionally robust to increases in the intermediary’s enforcement power.

While our characterization of CCE uniqueness focuses on pure equilibria, we show that allowing mixed outcomes changes little. A unique CCE can be mixed, but only in a near-degenerate sense: if the CCE is unique and not pure, then exactly two players randomize over two actions each, and the induced $2\times 2$ interaction is of matching-pennies type (Proposition~\ref{prop_unique_CCE_mixed}). Furthermore, in symmetric environments, uniqueness forces purity. 
Together, these observations suggest that mixed uniqueness is essentially a two-player artifact, ruled out whenever the environment exhibits symmetry or lacks a designated pair of players with a special strategic role.

Our results provide a toolkit for designing collusion-proof interactions. The key lesson is that ruling out coordinated deviations does not require engineering dominant strategies. To make a target profile $a^*$ the unique CCE, it is enough to ensure that any alternative profile generates sufficiently strong negative externalities. This boils down to checking one inequality, making uniqueness tractable.

We illustrate this logic in an application to contests, which are a workhorse framework for rent-seeking, R\&D races, litigation, political competition, and sports. In contests, a natural concern is that participants coordinate on low efforts and effectively share the prize. We demonstrate that such collusion is impossible in a stark sense for broad classes of two-player contests (Proposition~\ref{prop_contest}): whenever a contest admits a strict pure Nash equilibrium, that equilibrium is the unique CCE. We then exploit the characterization to show how canonical forms such as Tullock can be improved in terms of equilibrium effort and fairness toward low-effort participants, while preserving uniqueness of CCE.

\paragraph{Roadmap.} The rest of the paper proceeds as follows. In the remainder of this introduction, we discuss related literature. Section~\ref{sec:model} introduces the model of delegation and the benchmark solution concepts, CCE and IRCP. Section~\ref{sec:IRCP} characterizes games with a unique IRCP and shows that they are exactly those affinely equivalent to enforcement games. Section~\ref{sec:CCE} characterizes games with a unique CCE, establishes the link to IRCP through strategic equivalence, and describes applications to contests. Section~\ref{sec:extensions} discusses extensions and open questions.

\paragraph{Related literature.}
Our paper builds on and connects several strands of work on correlation, delegation, and equilibrium selection.  At a conceptual level, we study robustness of equilibrium predictions as one enlarges the set of feasible coordination devices. Nash outcomes sit inside the set of correlated equilibria (CE) \citep*{aumann1974subjectivity,aumann1987correlated}, and $\mathrm{CE}\subseteq\mathrm{CCE}\subseteq\mathrm{IRCP}$.  Our focus is the polar case in which these increasingly permissive solution concepts nevertheless collapse to a single outcome.  When that happens, the prediction becomes insensitive not only to equilibrium selection within Nash, but also to whether players can communicate, correlate, or delegate.

A large literature provides foundations for these solution concepts based on communication, learning, and long-run interaction.  CE captures outcomes of pre-play communication \citep*{barany1992fair,ben1998correlation,gerardi2004unmediated,lehrer1997one}; CCE is the target of no-external-regret learning \citep*{hannan1957approximation}; learning dynamics satisfying stronger notions such as no-internal-regret converge to CE \citep*{foster1997calibrated,hart2000simple,fudenberg1999conditional}.  In repeated games, equilibrium payoffs must satisfy individual rationality constraints \citep*{benoit1984finitely,fudenberg1986folk,aumann1994long} and thus empirical distributions of action profiles are captured by IRCP. Imperfect monitoring can shrink the outcomes that can be sustained; \cite{awaya2019communication} identify regimes in which the repeated play can be approximated by CCE rather than IRCP. \cite{denti2023robust} relate CCE to outcomes induced by rational inattention in games.
Taken together, these foundations imply that games with unique IRCP or CCE exhibit exceptional robustness: the same outcome emerges across communication protocols, learning rules, and long-horizon strategic environments.

Our motivating perspective differs from the classical ``benevolent mediator'' view: we take the intermediary to be potentially misaligned with a regulator and ask when the underlying strategic environment itself rules out harmful coordination.  This concern resonates with the growing literature on algorithmic and platform-induced coordination \citep*{calvano2020artificial,fish2024algorithmic,assad2024algorithmic}, as well as with the broader alignment perspective in which delegation to AI agents, whose objectives may not match those of the players, can shift outcomes without violating players’ individual rationality constraints \citep*{liang2025artificial,fudenberg2025friend,hadfield2019incomplete}.  It also connects to the literature on cartels and bidding rings, where intermediated coordination induces correlated play \citep*[e.g.,][]{mcafee1992bidding,ortner2024mediated}.  Our contribution is to provide a characterization of when such coordination opportunities fail to enlarge the outcome set, phrased in terms of the primitives of the strategic interaction.

 \citet*{moulin1978strategically} introduced CCE with a different motivation: to expand the implementable set relative to CE, and subsequent work established that the inclusion is strict in many environments \citep*{gerard1978correlation,ray2013coarse,moulin2014coarse,moulin2014improving,feldman2016correlated,dokka2023equilibrium}.  In parallel, the literature on learning in games discovered CCE as a set containing limits of no-regret dynamics \citep*{hannan1957approximation,hart2000simple,hart2003regret,hart2013simple}. Sometimes referred to as the Hannan set, the term ``coarse correlated equilibrium'' was popularized by  \cite*{young2004strategic} and \cite*{roughgarden2015intrinsic}.  

The IRCP benchmark is much less studied. While implicit in folk-theorem reasoning, it is rarely treated as a solution concept because of its apparent permissiveness \citep*[][Chapter~6]{myerson1991book}.  Our results show that, despite this permissiveness, IRCP can deliver sharp, non-knife-edge predictions in economically meaningful environments—and that the structure behind uniqueness links IRCP and CCE through strategic equivalence.
\cite*{csoka2024guaranteed} introduce a solution concept reminiscent of IRCP, called guaranteed utility equilibrium (GUE), which combines individual guarantees similar to those of IRCP with Pareto optimality. The relation to our results is discussed in Appendix~\ref{app_GUE}. Dynamic mechanism design applications of GUE are explored by \cite*{csoka2023collusion} and \cite*{csoka2025prior}.

General uniqueness results are scarce even for Nash equilibrium, and even scarcer for the more permissive solution concepts we study. Trivially, Nash, CE, and CCE are unique in games that admit a dominant action profile. Beyond dominance, CE is unique in potential games with a strictly concave potential and in related generalizations \citep*{neyman1997correlated,ui2008correlated,cao2025correlated}, as well as in generic zero-sum games \citep*{forges1990correlated,viossat2006geometry}; see also \cite*{kattwinkel2022mechanisms} and \cite*{kartik2024information}. Uniqueness of CE is also known to confer strong robustness to incomplete information and payoff perturbations \citep*{kajii1997robust,viossat2008having,einy2022strong}. A unique CE implementation of a social choice function is studied by \cite*{pei2025robust} and \cite*{banerjee2025correlated}. 

For CCE, uniqueness holds in the payoff space for zero-sum games \citep*{moulin1978strategically,macqueen2023proof} and symmetric non-atomic congestion games admitting a convex potential \citep*{koessler2025correlated}. Strictly socially concave games, introduced by \citet*{even2009convergence}, are a class in which no-regret dynamics converge to a single limiting payoff vector; \citet*{hart2015markets} link this convergence to the uniqueness of the CCE.
These results imply unique CCE in environments such as Tullock contests and linear Cournot oligopolies with strictly convex costs; further examples of games with a unique CCE can be found in \citet*{nadav2010no,baudin2023strategic,ahunbay2024uniqueness}.  Our approach recovers social concavity as a sufficient condition for uniqueness by linking it to enforcement games.  More broadly, our characterizations yield new tractable conditions for the uniqueness of CCE, and (because singleton CCE imply singleton CE and Nash) new sufficient conditions for uniqueness of CE and Nash that do not rely on dominant strategies.

Finally, on the technical level, our paper contributes to understanding how duality and extreme-point arguments can be used to derive structural properties of solutions in games. In this respect, we broaden the scope of extreme-point methods that have recently become prominent in economic theory, e.g., \citep*{kleiner2021extreme,nikzad2022constrained,arieli2023optimal,kleiner2024extreme,lahr2024extreme,yang2024monotone,rudov2025extreme,augias2025economics,yang2025multidimensional,kleiner2026extreme}.

\section{Games and Solutions}
\label{sec:model}

A normal-form game is a tuple $G=(N,(A_i)_{i\in N},(u_i)_{i\in N})$, where $N=\{1,\ldots,n\}$ is a finite set of players, each player $i\in N$ has an action set $A_i$, and $u_i\colon A\to\R$ is a von Neumann--Morgenstern utility function defined on the set of action profiles $A=\times_{i\in N}A_i$. For each $i\in N$, write $A_{-i}=\times_{j\neq i}A_j$. When no confusion arises, we suppress $N$ and $(A_i)_{i\in N}$ and write $G=(A,u)$.

Unless stated otherwise, all games are finite, meaning that each $A_i$ is a finite set. However, most of our results also extend, at least in part, to games with a continuum of actions. We discuss such extensions separately. To avoid trivialities, we also assume $|A_i|\ge 2$ for all $i\in N$.

Let $\Delta(A)$ denote the probability simplex over $A$. A distribution $\mu\in\Delta(A)$ represents the outcome selected by an intermediary that correlates players' actions; when $\mu$
 is a product distribution, it can equivalently represent a decentralized outcome in which players mix independently.

Our central solution concept for modeling delegation is \textbf{Coarse Correlated Equilibrium (CCE)} \citep{hannan1957approximation,moulin1978strategically}. A distribution $\mu \in \Delta(A)$ is a CCE if, for every player $i \in N$ and every deviation $a_i' \in A_i$,
$$
        \underbrace{\sum_{a \in A} \mu(a) u_i(a)}_{\text{utility of delegation}} \ge \underbrace{\sum_{a \in A} \mu(a) u_i(a_i', a_{-i})}_{\text{utility of a fixed action $a_i'$}}    
$$
CCE captures the scenario where the intermediary is weak and cannot punish players who opt out of delegation. Accordingly, when player $i$ considers deviating, she assumes that the intermediary does not react to her deviation and that others keep playing what they would have played absent the deviation; however, she loses information about the realized action profile.
CCE is exactly the set of correlated action profiles for which no player would find such a deviation profitable.

An alternative way to interpret CCE is to assume that the intermediary recommends  actions to players rather than taking actions on their behalf,  but players
commit \textit{ex ante} to follow these recommendations before observing them. 
It is instructive to compare this interpretation with that of classical \textbf{Correlated Equilibrium (CE)} \citep{aumann1974subjectivity}, which captures situations where each recommendation can be reconsidered \textit{ex post}.
A distribution $\mu\in\Delta(A)$ is a CE if, for every player $i\in N$ and every pair of actions $a_i,a_i'\in A_i$,
$$
\underbrace{\sum_{a_{-i}\in A_{-i}}\mu(a_i,a_{-i})\,u_i(a_i,a_{-i})}_{\text{utility of playing $a_i$ whenever it is recommended}}
\ge \underbrace{\sum_{a_{-i}\in A_{-i}}\mu(a_i,a_{-i})\,u_i(a_i',a_{-i})}_{\text{utility of playing $a_i'$ whenever $a_i$ is recommended}}.
$$

We next define a permissive benchmark concept that corresponds to the opposite extreme of the intermediary's power: when player $i$ deviates, the intermediary can coordinate all other players on an adversarial response that minimizes $i$'s payoff. Following \citet[Chapter 6]{myerson1991book}, a distribution $\mu\in \Delta(A)$ is an \textbf{Individually Rational Correlated Profile (IRCP)} if 
$$
        \underbrace{\sum_{a \in A} \mu(a) u_i(a)}_{\text{utility of delegation}} \ge \underbrace{\max_{\nu\in \Delta(A_i)}\min_{a_{-i}}\sum_{a_i \in A_i} \nu(a_i) u_i(a_i, a_{-i})}_{\text{what $i$ can guarantee against adversarial play}}    
$$
That is, each player's expected utility must be at least their individually rational maximin level reminiscent of the classical folk theorem.\footnote{To avoid confusion, we emphasize that the folk theorem characterizes payoff vectors that can arise in repeated interaction, whereas our object of interest is the induced distribution over action profiles. In particular, the same individually rational payoff vector may correspond to multiple action distributions $\mu$, each of which we classify as IRCP.
Another distinction concerns the definition of the guarantee level. By the minimax theorem, the term $\max_{\nu\in\Delta(A_i)}\min_{a_{-i}\in A_{-i}}$ in the definition of $\IRCP$ can equivalently be written as $\min_{\tau\in\Delta(A_{-i})}\max_{a_i\in A_i}$, reflecting that the intermediary can correlate the actions of the punishing players. By contrast, correlated punishments cannot be sustained in the folk theorem under standard assumptions (no communication or mediation, no common randomization devices, perfect monitoring); see, e.g., \cite*{fudenberg1986folk,horner2006folk}.}

The individual rationality benchmark is often viewed as too permissive to qualify as a solution concept and is instead seen as a constraint on outcomes that require further refinement. Nevertheless, we will show that it delivers sharp predictions in some economically relevant environments.
Importantly, IRCP allows the intermediary to coordinate punishments even when carrying them out hurts the non-deviators. Adding feasibility constraints on punishments would weaken the intermediary,  shrinking the implementable set relative to $\IRCP$ and resulting in intermediate solution concepts between $\IRCP$ and $\CCE$.\footnote{For example, a classical modeling choice to ensure that punishments do not violate non-deviators' individual rationality constraint, going back to the folk theorem of \cite*{friedman1971non}, is to punish deviations by switching to a Nash equilibrium of the stage game. This approach is common in work on cartels and collusion in auctions, e.g., \cite*{RACHMILEVITCH20131714,ortner2024mediated,nakabayashi2025scoring}.} While our analysis focuses on $\IRCP$ and $\CCE$ as the two extremes of intermediary power, our results also speak to such intermediate concepts, since uniqueness of $\IRCP$ implies their uniqueness.

\medskip
For a given game $G$, we denote the sets of equilibria by $\CCE(G)$, $\CE(G)$, and $\IRCP(G)$, respectively. These sets are defined by linear inequalities and hence are convex polytopes contained in $\Delta(A)$. All these polytopes are non-empty for finite games, as they contain the set of \textbf{Nash equilibria} $\NE(G)$. 
The solution concepts satisfy the chain of inclusions
$$
\NE(G)\subseteq \CE(G)\subseteq \CCE(G)\subseteq \IRCP(G).
$$
Furthermore, any product distribution $\nu=\nu_1\times\cdots\times\nu_n$  contained in $\CCE(G)$ (and hence in $\CE(G)$) is necessarily a Nash equilibrium. In contrast, $\IRCP(G)$ may contain additional product distributions such as the product of maximin strategies.

\smallskip

To illustrate the distinction between CE, CCE, and IRCP, consider the classic rock-paper-scissors game. As is generically true for zero-sum games \citep{forges1990correlated}, CE is unique, coinciding with the unique Nash equilibrium where both players uniformly randomize independently of each other. CCE is more expansive, including any correlated action profile with uniform marginals in which each player receives an expected utility of zero. For example, the uniform distribution over (rock, rock), (paper, paper), and (scissors, scissors) constitutes a CCE. IRCP is the most permissive solution concept, requiring only that each player receive zero utility in expectation. This includes, for example, the distribution that assigns probability $1/2$ to the two profiles (rock, paper) and (paper, rock).

\medskip

The notions of CCE and IRCP extend straightforwardly to games with continuum action spaces, provided action sets are measurable and utilities are measurable and bounded. In the inequalities defining these solution concepts, summation is replaced with integration and the maximin level of IRCP is replaced with the $\sup\inf$ level.
The non-emptiness of these solution concepts is guaranteed under standard assumptions of compactness and continuity, which in particular imply the existence of a mixed Nash equilibrium. While we primarily focus on finite games, we will discuss some extensions and examples pertaining to the continuous environment within the following sections.

\section{Unique Individually Rational Correlated Profiles}\label{sec:IRCP}

We start by considering our most permissive solution concept, IRCP. Accordingly, requiring that $\IRCP(G)$ is unique imposes the most stringent restrictions on the strategic environment. The intuitions developed here will be useful for addressing the uniqueness questions for CCE in the next section.

Since $\NE(G)\subseteq \CE(G)\subseteq \CCE(G)\subseteq \IRCP(G)$, the uniqueness of IRCP immediately implies the uniqueness of all the other equilibrium sets.
Accordingly, the decentralized Nash outcome in such games is robust to rich informational and institutional environments: allowing correlation via a mediator (CE) or delegation (CCE) does not enlarge the predicted outcome set when $\IRCP(G)$ is a singleton. The same robustness carries over to standard dynamic foundations---no-regret learning (external regret for CCE and internal regret for CE) converges to the corresponding equilibrium sets, and long-horizon repeated-game equilibria must be individually rational---so uniqueness of $\IRCP(G)$ pins down these predictions to the same Nash outcome as well.
\smallskip

To get intuition about what constraints the uniqueness of IRCP imposes on a game, consider two benchmark games with mixed Nash equilibria: matching pennies and rock--paper--scissors. In both games, each player's individually rational level equals the value of the underlying two-player zero-sum game, which under the standard normalization is $0$. Hence, any distribution $\mu$ that yields expected payoff $0$ to both players belongs to $\IRCP(G)$. The set of such distributions contains mixed Nash equilibria as well as other distributions. For example, in either game, the distribution in which one player mixes uniformly while the other plays a pure action yields expected payoff $0$ to both players and therefore lies in $\IRCP(G)$. This non-uniqueness is closely tied to the fact that equilibrium play is mixed: the individually rational constraints pin down only the value $0$, leaving many action distributions compatible with it.

As the following proposition shows, a unique IRCP must necessarily be a pure Nash equilibrium. Moreover, each player's equilibrium action must be the unique best response to the actions of others. Formally, an action profile $a^*\in A$ is a \textbf{strict} Nash equilibrium if $u_i(a^*)>u_i(a_i,a_{-i}^*)$ for all~$i$ and~$a_i\ne a_i^*$.
\begin{proposition}\label{pr:IRCPpure}
    If $\IRCP(G)$ is a singleton, then it is a pure Nash equilibrium. Furthermore, it is strict.
\end{proposition}
The proposition is based on the following simple structural insight. $\IRCP(G)$ always contains a canonical outcome: the product distribution where every player uses their maximin strategy. If $\IRCP(G)$ is a singleton, this distribution must equal the unique Nash equilibrium, which drives the result.

\begin{proof}
 Let $\nu_i$ be player $i$'s maximin strategy. Then $\nu=\nu_1\times \cdots \times \nu_n$ is an IRCP.
 We now show that $\nu$ is a pure Nash equilibrium which is, moreover, strict. 
 Let $a_i$ be an action giving player $i$ the highest expected utility against $\nu_{-i}$. Then $\nu'=\delta_{a_i}\times\nu_{-i}$ is also an IRCP since all $j\ne i$ play their maximin strategies and thus are above their maximin level and player $i$ is also above this level by the choice of $a_i$. Using the fact that IRCP is a singleton, we conclude that $\nu_i=\delta_{a_i}$. Thus $\nu$ is a point mass. Finally, $a_i$ must be the unique best response to $\nu_{-i}$, so $\nu$ is a strict Nash equilibrium. 
\end{proof}

Proposition~\ref{pr:IRCPpure} implies that, in our pursuit of games with a unique IRCP, we must focus on those that have a unique Nash equilibrium which is pure. A natural class of games with this property is that of dominance-solvable games, where sequential elimination of strictly dominated actions leads to a single action profile.  
However, even a much stronger property---the existence of a dominant action for each of the players---which ensures that $\NE(G)=\CE(G)=\CCE(G)$ are singletons, does not guarantee the uniqueness of IRCP.

For example, consider the prisoner's dilemma. The action $d$ (defect) is dominant for both players, but $\IRCP$ is not a singleton. To see this, note that the set $\IRCP(G)$ is closed under Pareto improvements for any game $G$. Since full cooperation $(c,c)$ Pareto improves over the Nash outcome $(d,d)$ in the prisoner's dilemma, we conclude that $(d,d)$, $(c,c)$, as well as their mixtures lie in $\IRCP$.

This example might convince the reader that the uniqueness of IRCP is even more demanding than the existence of dominant actions. 
However, as we will see, IRCP can be unique even in games that possess neither dominant actions nor dominated ones.
\smallskip

Consider the following example. There is a circular road with paid parking. Parking yields utility of $v$ and costs $c>0$. Two drivers decide whether to pay or to park without paying. If a driver decides not to pay, she enjoys the utility of $v$ but may be caught, in which case she pays a ticket of amount $t>c$. The probability of being caught depends on the location  $L=\{l_1,\ldots, l_m\}$, $m\geq 3$, on the loop at which they park. An inspector inspects the parking locations counterclockwise along the loop, starting from a uniformly random location. If both drivers decide not to pay, only the one the inspector finds first gets a ticket (upon seeing this, the other driver drives away). If the drivers choose the same location, each gets a ticket with probability $1/2$. If only one driver parks illegally, she gets a ticket for sure.

To get some intuition about the incentives in this game, suppose that both drivers decide not to pay for parking. If one of them chooses the location $l_i$, the other one minimizes the chance of being caught by parking at $l_{i+1}$ (with the standard cyclic index convention). That is, each violator wants to park right after the other along the inspector's way. If the first driver parks at $l_i$ and the second parks at $l_j$ with $j\ne i$, then the probability that the second driver is caught is given by $\frac{(j-i) \mod m}{m}$, which is exactly the probability that the inspector's starting point is in the range between $l_{i+1}$ and $l_j$. The resulting expected payoffs are presented in Figure~\ref{fig:ticketing} for $m=3$ locations.

The pair of actions (pay, pay) is always a Nash equilibrium of this game. However, paying for parking is not a dominant action unless the ticket amount is so big that even parking next to the other violator and thus receiving a ticket with probability of only $1/m$ gives a lower utility than paying. Indeed, paying becomes dominant for $t>m\cdot c$. We now argue that IRCP is a singleton even for much smaller values of $t$.

Indeed, the maximin level is $v-c$, which each player guarantees by paying. For $t>2c$, in any profile other than (pay, pay) the sum of utilities is less than $2(v-c)$, meaning that at least one player receives less than their maximin level. Similarly, in any correlated profile placing positive weight outside of (pay, pay), the expected welfare is below $2(v-c)$ and thus at least one player is below their maximin level. Hence, no such correlated profile is an IRCP. Accordingly, for $t\in (2c,\,mc]$, paying is the unique NE, CE, CCE, and IRCP, while not being a dominant action profile. 

\begin{figure}[t]
\centering
\begin{tikzpicture}[font=\small]
\matrix (G) [matrix of math nodes,
  nodes in empty cells,
  nodes/.append style={draw, align=center, minimum height=9mm, minimum width=35mm, inner sep=2pt},
  row sep=\pgflinewidth, column sep=\pgflinewidth
]{
  & \text{pay} & l_1 & l_2 & l_3\\
  \text{pay}
    & (v-c,\,v-c) & (v-c,\,v-t) & (v-c,\,v-t) & (v-c,\,v-t)\\
  l_1
    & (v-t,\,v-c)
    & (v-\frac{t}{2},\,v-\frac{t}{2})
    & (v-\frac{2t}{3},\,v-\frac{t}{3})
    & (v-\frac{t}{3},\,v-\frac{2t}{3})\\
  l_2
    & (v-t,\,v-c)
    & (v-\frac{t}{3},\,v-\frac{2t}{3})
    & (v-\frac{t}{2},\,v-\frac{t}{2})
    & (v-\frac{2t}{3},\,v-\frac{t}{3})\\
  l_3
    & (v-t,\,v-c)
    & (v-\frac{2t}{3},\,v-\frac{t}{3})
    & (v-\frac{t}{3},\,v-\frac{2t}{3})
    & (v-\frac{t}{2},\,v-\frac{t}{2})\\
};

% --- Separating lines between labels and payoffs ---
% Vertical separator: between column 1 (row labels) and column 2 (payoffs)
\coordinate (Vtop) at ($(G-1-1.north east)!0.5!(G-1-2.north west)$);
\coordinate (Vbot) at ($(G-5-1.south east)!0.5!(G-5-2.south west)$);
% keep x from Vbot, y from both  -> perfectly vertical
\draw[line width=0.9pt] (Vbot) -- (Vbot |- Vtop);

% Horizontal separator: between row 1 (column labels) and row 2 (payoffs)
\coordinate (Hleft)  at ($(G-1-1.south west)!0.5!(G-2-1.north west)$);
\coordinate (Hright) at ($(G-1-5.south east)!0.5!(G-2-5.north east)$);
\draw[line width=0.9pt] (Hleft) -- (Hright);

% --- Column player's brace over l1,l2,l3 (top) ---
\draw[decorate,decoration={brace,amplitude=4pt,raise=2pt}]
  (G-1-3.north west) -- (G-1-5.north east)
  node[midway,above=6pt] {\scriptsize not pay and park at};

% --- Row player's "rotated overbrace" over l1,l2,l3 (left, vertical) ---
% Draw a brace spanning rows l1..l3 in the first column, outside to the left,
% and rotate the text by 90 degrees.
\draw[decorate,decoration={brace,mirror,amplitude=4pt,raise=2pt}]
  (G-3-1.north west) -- (G-5-1.south west)
  node[midway,xshift=-10pt,rotate=90] {\scriptsize not pay and park at};

\end{tikzpicture}
\caption{Expected payoff matrix for $m=3$ locations. Entries are $(u_1,u_2)$.}
\label{fig:ticketing}
\end{figure}

The IRCP uniqueness that we observed in this game is robust to the details of the game. One can easily modify the construction to allow for $n\geq 2$ drivers, non-uniform randomization by the inspector, catching more than one violator, or a continuum of locations.

The phenomenon we observed in this example is common to a broader class of $n$-person games that we now describe.
\begin{definition}\label{def:enforcement}
   We call $H=(A,v)$ an \emph{enforcement game} if there is an action profile $a^*\in A$ such that
   $$v_i(a^*)=0,\qquad v_i(a_i^*,a_{-i})\geq 0 \quad \text{for all $i$ and $a_{-i}$} \qquad\text{and}\qquad \sum_i v_i(a)<0\quad \text{for all $a\ne a^*$}.$$
   That is, each player's utility at $a^*$ is zero, each player $i$ can guarantee this utility level unilaterally by playing $a^*_i$, and $a^*$ is the only profile maximizing utilitarian welfare. We call $a^*$ the \emph{self-enforcing profile}. 
\end{definition}

The self-enforcing profile $a^*$ is a strict Nash equilibrium. At $a^*$ every player receives $0$, while any unilateral deviation $a=(a_i,a_{-i}^*)$ with $a_i\ne a_i^*$ yields a strictly negative payoff to the deviator: the total welfare satisfies $\sum_j v_j(a)<0$, yet each non-deviating player $j\ne i$ contributes $v_j(a)\ge 0$, so the entire welfare shortfall is borne by player~$i$.

Mimicking the argument from the parking example, we obtain the following lemma.
\begin{lemma}\label{lm:enf_unique}
The self-enforcing profile $a^*$ is a unique $\IRCP$.
\end{lemma}
\begin{proof}
    Since $a^*$ is a Nash equilibrium, it is an $\IRCP$. It remains to show that any $\mu\in \Delta(A)$ placing positive weight on $A\setminus \{a^*\}$ cannot be an IRCP. The expected welfare induced by such a $\mu$ is negative. Thus there is a player $i$ with negative expected utility. However, this player can guarantee a payoff of at least $0$ by playing $a_i^*$. Hence, this player is below their maximin level at $\mu$ and thus $\mu$ is not an IRCP.
\end{proof}

One can construct new games with a unique IRCP by noticing that the set of $\IRCP$ is preserved under affine equivalence of games.
Formally, consider games $G=(A,u)$ and  $H=(A,v)$ sharing the same set of action profiles and players. These games are \textbf{affinely equivalent} if $v_i(a)=\gamma_i (u_i(a) +\beta_i) $ for some constants $\gamma_i>0$ and $\beta_i\in \R$. For such games, $\IRCP(G)=\IRCP(H).$ That is, the solution concept is not sensitive to the choice of von Neumann-Morgenstern utility representing the underlying preference over lotteries. 

Accordingly, a game affinely equivalent to an enforcement game also has a unique IRCP. It turns out that no other game has this property.

\begin{theorem} \label{th:IRCP}
    A game $G$ has a unique $\IRCP$ if and only if $G$ is affinely equivalent to an enforcement game.
\end{theorem}

\begin{proof}
    By Lemma~\ref{lm:enf_unique}, a game affinely equivalent to an enforcement game has a unique IRCP. So we only need to prove the other direction of the theorem.

    Consider a game $G=(A,u)$ with a unique IRCP and maximin levels $\underline{u}_i$. By Proposition~\ref{pr:IRCPpure}, this IRCP is a pure Nash equilibrium, where some action profile $a^*$ is played. Since the product of maximin strategies is also an IRCP, it coincides with $a^*$ by the uniqueness assumption. Thus  $a_i^*$ is $i$'s maximin strategy.
    Moreover, $\underline{u}_i=u_i(a^*)$. Indeed, if $\underline{u}_i<u_i(a^*)$, then player~$i$ will still be above their maximin level by mixing $a_i^*$ with some other action $a_i\ne a_i^*$ with weights $1-\varepsilon$ and $\varepsilon$ where $\varepsilon>0$ is small enough. Since other players keep playing their maximin strategies, they remain above their maximin levels. Thus the constructed distribution is an IRCP, contradicting its uniqueness.
    
    We now show that a positive affine transformation of utilities makes $G$ an enforcement game. The weights in this transformation arise from optimal strategies in an auxiliary zero-sum game.    

    Consider the following zero-sum game played between Maximizer and Minimizer. Maximizer selects an action profile $a\in A\setminus \{a^*\}$ and Minimizer selects a player $i\in N$. The Maximizer aims to maximize $u_i(a)-u_i(a^*)$. Since $G$ has no IRCP supported on $A\setminus \{a^*\}$, whatever mixed-strategy Maximizer chooses, Minimizer can respond by selecting a player receiving below her maximin level. %That is, Minimizer can force a negative payoff against any given strategy of the Maximizer. 
    By the minimax theorem, there is a mixed-strategy of Minimizer, $\gamma \in \Delta(N)$, that forces a negative payoff to Maximizer no matter what strategy Maximizer chooses:
    \begin{equation}\label{eq_minimizer}
        \sum_i \gamma_i\Big(u_i(a) - u_i(a^*)\Big)<0\qquad \text{for all $a\neq a^*$}.
    \end{equation}
Now plug in $a=(a_i,a_{-i}^*)$ with $a_i\neq a_i^*$. 
The terms with $j\ne i$
 are nonnegative since $a_j=a_j^*$
 is player $j$'s maximin strategy and thus guarantees $j$ at least $\underline{u}_j=u_j(a^*)$. Hence, the strict inequality can only be satisfied if $\gamma_i\Big(u_i(a)-u_i(a^*)\Big)<0$. Thus the weight  $\gamma_i>0$ for all~$i$.

    Consider the game $H=(A,v)$ with utilities given by  
    $$v_i(a)=\gamma_i\left(u_i(a)-u_i(a^*)\right).$$ 
    By~\eqref{eq_minimizer}, for any profile $a\neq a^*$, $ \sum_i v_i(a)<0$. Finally, since $a_i^*$ is player~$i$'s maximin action in $G$, for all $a_{-i}$, $u_i(a_i^*,a_{-i})\ge u_i(a^*)$, so $v_i(a_i^*,a_{-i})\ge 0$ and $v_i(a^*)=0$. We conclude that $H$ is an enforcement game  affinely equivalent to $G$.
\end{proof}

An implication of Theorem~\ref{th:IRCP} is that a unique IRCP must be a strong Nash equilibrium, meaning that any multilateral deviation from the equilibrium cannot strictly benefit all of the deviators. Indeed, any enforcement game has this property. Curiously, this connection between IRCP uniqueness and strong Nash equilibrium means that a demanding notion of robustness to unilateral deviations implies a form of stability against coalitional ones as well.
In fact, enforcement games exhibit an even stronger stability notion referred to as strict strong Nash equilibrium \citep{chwe1994farsighted}, 
requiring that, for any joint deviation, at least one deviator must be harmed.
Compared to strict strong Nash equilibria, the condition for a unique IRCP is more restrictive, since it further requires that a multilateral deviation does not harm any non-deviators.

It is also worth comparing games with a unique IRCP to those with a dominant action profile. Neither one of these conditions imply the other, as evidenced by the prisoner's dilemma, which has a unique dominant profile but a multiplicity of IRCP, and the parking game in Figure~\ref{fig:ticketing} with $t\in (2c,mc)$, which has a unique IRCP but no dominant profiles. The insight shed by the parking game applies more generally to any setting where a designer wants to set penalties to ensure that players cannot deviate from a desired action even when they are prone to delegate to a strong intermediary. As the parking example shows, the designer only needs to set penalties high enough such that at least one deviator would regret their deviation. For example, when it is challenging to penalize many deviators, a significant enough penalty applied to one scapegoat is enough to enforce the desired profile.

\paragraph{Continua of actions.} So far, all of our results and discussions have been formulated in the setting where each player has a finite set of actions. However, many of these insights extend to settings in which players have a continuum of actions. 

Our first result, Proposition~\ref{pr:IRCPpure}, which stated that a unique IRCP is a pure and strict Nash equilibrium, extends to the continuous action setting. Indeed, the proof of that proposition holds as stated for games with a continuum of actions as long as maximin strategies and best responses are well-defined. For example, it is enough to require that $A_i$ are compact sets and $u_i(a_i,a_{-i})$ is upper semicontinuous in $a_i$ and lower semicontinuous in $a_{-i}$. There are two aspects of this extension worth highlighting.  First, the proposition implies the existence of a Nash equilibrium under these mild regularity assumptions and uniqueness of IRCP.
Second, while the continuous setting is prone to creating indifferences, finding a pure Nash equilibrium that is \emph{not strict} is still enough to conclude that IRCP is not a singleton.

We also note that the definition of an enforcement game, as well as Lemma~\ref{lm:enf_unique} establishing uniqueness of IRCP for enforcement games, extend to continuum action spaces. The proof of Lemma~\ref{lm:enf_unique} goes through verbatim for arbitrary measurable spaces $A_i$ and measurable utilities $v_i\colon A\to\R.$ 

Finally, a version of Theorem~\ref{th:IRCP} holds for symmetric games with a continuum of actions. 
\begin{proposition}\label{prop_continuum_IRCP}
    Consider a game $G$ with measurable action sets and bounded measurable utilities. The following assertions hold:
\begin{itemize}
    \item If $G$ is affinely equivalent to an enforcement game, then it has a unique $\IRCP$ which is a pure action profile $a^*$.
    \item If $G$ is symmetric and each player has a maximin strategy, the converse is also true.
\end{itemize}
\end{proposition}

This proposition is proved in Appendix~\ref{app_continuum_IRCP}. The symmetry assumption in the second item obviates the need for the duality argument used to ensure strict inequality for the welfare at~$a^*$ and at other profiles in the proof of Theorem~\ref{th:IRCP}. Without symmetry, strict inequality in the continuous setting is not guaranteed. The assumption of the existence of a maximin strategy in the second item can be replaced with compactness of $A_i$ and upper semicontinuity of $u_i$ for all $i$.

\section{Unique Coarse Correlated Equilibria}\label{sec:CCE}

In this section, we study games that have a unique CCE. Since CCE is contained in IRCP, any game with a unique IRCP automatically has a unique CCE.

Enforcement games defined in the previous section (Definition~\ref{def:enforcement}) have a unique IRCP. Since IRCP is invariant to positive affine transformations, applying them to enforcement games again preserves uniqueness. CCE is invariant under a broader class of transformations, which immediately gives a class of games with unique CCE that is larger than that with unique IRCP.

Consider two games $G=(A,u)$ and $H=(A,v)$ sharing the same set of action profiles and players. These games are \textbf{strategically equivalent} if 
$v_i(a_i,a_{-i})=\gamma_i\left(u_i(a_i, a_{-i})+\beta_i(a_{-i})\right)$ where $\gamma_i>0$ and  $\beta_i\colon A_{-i}\to \R$. That is, the utilities in $H$ can be obtained from the utilities in $G$ by rescaling and adding a function that depends only on the actions of others to each player's utility. In contrast to positive affine transformations, the additive perturbation $\beta_i$ may not be constant. 

The set of IRCP is not preserved under strategic equivalence since adding a function of others' actions to $i$'s utility may affect the ability of others to punish $i$ for deviations. By contrast, CCE predictions agree for strategically equivalent games: $\CCE(G)=\CCE(H)$. Indeed,
only utility differences $u_i(a_i,a_{-i})-u_i(a_i',a_{-i})$ are relevant for $i$'s incentives in a CCE, and such differences are preserved up to scale.
In fact, this broader invariance is what singles out $\CCE$ within the set of $\IRCP$.
\begin{lemma}\label{lm:CCE_SE_IRCP}
A distribution  $\mu$ belongs to $\CCE(G)$ if and only if $\mu$ belongs to  $\IRCP(H)$ for all $H$ strategically equivalent to $G$.
\end{lemma}
In other words, $\CCE$ is the set of all correlated outcomes that satisfy individual rationality and respect strategic equivalence.\footnote{ The proof shows that it is enough to require strategic equivalence with respect to a very particular set of transformations that make the utility of a certain action $a_i'$ exactly $0$ for a given player $i$.}
\begin{proof}
    Consider a game $G=(A,u)$ and $\mu\in \CCE(G)$. Since $\CCE$ respects strategic equivalence, $\mu\in \CCE(H)$ for all strategically equivalent $H$. Thus $\mu\in \IRCP(H)$ since $\CCE(H)$ is a subset of $\IRCP(H)$.
    Conversely, for $\mu\notin \CCE(G)$, we show that $\mu\notin \IRCP(H)$ for some strategically equivalent $H$. Since $\mu \notin \CCE(G)$, there is a player $i$ and an action $a_i'$ such that $\sum_{a} \mu(a) u_i(a) < \sum_{a} \mu(a) u_i(a_i', a_{-i})$. Consider $H=(A,v)$ with $v_j=u_j$ for $j\ne i$ and $v_i(a)=u_i(a)-u_i(a_i', a_{-i})$. Then player $i$ in $H$ can guarantee a utility of $0$ by playing $a_i'$. On the other hand,  $\sum_{a} \mu(a) v_i(a)<0$ and thus $\mu\notin \IRCP(H)$.     
\end{proof}

Enforcement games have a unique IRCP and thus games strategically equivalent to enforcement games have unique CCE.
In fact, such games exhaust all those with unique CCE that are pure.
\begin{theorem}\label{th:CCE}
    A game $G$ has a unique $\CCE$ that is pure if and only if $G$ is strategically equivalent to an enforcement game.
\end{theorem}

To get more intuition about what the strategic equivalence condition in Theorem~\ref{th:CCE} means, we express it in terms of the utilities of the original game $G=(A,u)$. This will help us with general structural insights as well as with applications.
To establish the equivalence, we must find weights $\gamma_i>0$ and additive perturbations $\beta_i(a_{-i})$ such that the game with utilities  $v_i(a_i,a_{-i})=\gamma_i\left(u_i(a_i, a_{-i})+\beta_i(a_{-i})\right)$ is an enforcement game. That is,
$u_i(a^*)+\beta_i(a_{-i}^*)=0$, $u_i(a_i^*,a_{-i})+\beta_i(a_{-i})\geq 0$, and $\sum_i\gamma_i \left(u_i(a)+\beta_i(a_{-i})\right)<0 $ for $a\ne a^*$. By taking the pointwise minimal $\beta_i$ permitted by the second condition, we satisfy the first condition automatically and only improve the last one. Thus it is without loss of generality to consider $\beta_i(a_{-i})=-u_i(a_i^*,a_{-i})$. Consequently, $G$ is strategically equivalent to an enforcement game if and only if the following inequality holds for some  weights $\gamma_i>0$ and all $a\ne a^*$
\begin{equation}\label{eq_CCE_check}
\Phi(a)<0\qquad \text{where}\qquad  \Phi(a)=\sum_i \gamma_i \Big(u_i(a_i,a_{-i})-u_i(a_i^*,a_{-i})\Big).
\end{equation}
Thus, checking the equivalence boils down to finding the weights $\gamma_i$. For symmetric games $G$, even this step is not needed, as one can always take $\gamma_i=1$ for all $i$. Indeed, if~\eqref{eq_CCE_check} is satisfied with some weights $\gamma_i$, the symmetry of the game implies that it is also satisfied with permuted weights. Averaging over all permutations, we conclude that it is satisfied by uniform weights.

The condition~\eqref{eq_CCE_check} can be interpreted as follows. Consider an action profile $a=(a_i,a_{-i}^*)$ corresponding to a unilateral deviation by player~$i$ from $a^*$. For each $j\ne i$, we have $u_j(a_j,a_{-j})-u_j(a_j^*,a_{-j})=0$ and, hence,
\begin{equation}\label{eq_potential}
    \Phi(a_i,a_{-i}^*)=\gamma_i\big(u_i(a_i,a_{-i}^*)-u_i(a^*)\big).
\end{equation}
That is, the value of $\Phi$ under a unilateral deviation is proportional to the change in the utility of the deviator. Accordingly, one can think of $\Phi$ as a potential representing deviations from $a^*$. We will refer to it as the \textbf{local potential} at $a^*$. If $a^*$ is a strict Nash equilibrium, the right-hand side of~\eqref{eq_potential} is strictly negative for $a_i\ne a_i^*$. Thus the inequality~\eqref{eq_CCE_check} is automatically satisfied at a strict Nash equilibrium for unilateral deviations $a=(a_i,a_{-i}^*)$.
Thus the gap between a strict Nash and a unique CCE is whether multilateral deviations lead to $\Phi<0$.
This can be seen as a requirement that multilateral deviations induce strong enough negative externalities.

\smallskip 

Comparing Theorems~\ref{th:IRCP} and~\ref{th:CCE} highlights an a priori unexpected connection between games with unique $\CCE$ and unique $\IRCP$. Indeed, a game has a unique pure $\CCE$ if and only if it is strategically equivalent to a game with a unique $\IRCP$.
Furthermore,~\eqref{eq_CCE_check} allows one to focus on a particular form of strategic equivalence. Thus $a^*$ is the unique CCE of $G=(A,u)$ if and only if $a^*$ is a unique $\IRCP$ in the game $G'=(A,u')$ with 
\begin{equation}\label{eq_CCE_to_IRCP_reduction}
u_i'(a)=u_i(a)-u_i(a_i^*,a_{-i}).
\end{equation} 
From a design perspective, this identity shows how one can modify payoffs in a game with a unique CCE to ensure that the equilibrium outcomes remain unchanged even if an intermediary can punish deviators.

The proof of the theorem is contained in Appendix~\ref{app_CCE}. 
There we  show  directly that $a^*$ is the unique CCE of $G=(A,u)$ if and only if it is a unique IRCP in the game $G'=(A,u')$ with $u'$ given in~\eqref{eq_CCE_to_IRCP_reduction}. 
Once this equivalence is established, Theorem~\ref{th:CCE} becomes a corollary of Theorem~\ref{th:IRCP}.

\smallskip

Theorem~\ref{th:CCE} implies robustness of games with unique CCE to small perturbations. Since in enforcement games Nash equilibria are strict, a unique pure $\CCE$ is strict as well. Furthermore, a small perturbation of a game with a unique CCE $a^*$ has the same unique CCE. Indeed, if $G=(A,u)$ satisfies condition~\eqref{eq_CCE_check} for some $\gamma$, the same strict inequality continues to hold for all games with utilities sufficiently close to $u$.
\begin{corollary}\label{cor_CCE_open}
The set of games with a unique pure $\CCE$ is open.
\end{corollary}
By the result of \cite*{viossat2008having}, the property of having a unique CE exhibits similar robustness to small perturbations. Interestingly, the situation for IRCP is more subtle. While there are games with unique IRCP that are robust to perturbations, there are those that are not. For example, by slightly perturbing an enforcement game $H=(A,v)$ such that all the inequalities in Definition~\ref{def:enforcement} are strict, we get a game affinely equivalent to an enforcement one and thus having a unique IRCP.
However, if $v_i(a_i^*,a_{-i})\equiv 0$, then any perturbation lowering $v_i(a_i^*,a_{-i})$ for some $a_{-i}\ne a_{-i}^*$ will result in a game with multiple IRCPs. The underlying reason is that, by lowering $v_i(a_i^*,a_{-i})$, we move the maximin level of player $i$ below $v_i(a^*)$ so that their individual rationality constraint no longer binds, opening up the possibility of other IRCPs described in the proof of Theorem~\ref{th:IRCP}.

\medskip

The extra flexibility offered by strategic equivalence leads to many more games with a unique CCE compared to those with a unique IRCP. 
For example, if $G=(A,u)$ admits a dominant action $a_i^*$ for each player, then $G$ is strategically equivalent to an enforcement game. Indeed, define $v_i(a)=u_i(a)-u_i(a_i^*,a_{-i})$. Then $v_i(a_i^*,a_{-i})\equiv 0 $ and $\sum_i v_i(a)<0$ for all $a\ne a^*$, since this inequality holds for each player $i$ not playing $a_i^*$ separately as $u_i(a_i^*,a_{-i})>u_i(a_i,a_{-i})$ by dominance.

 The class of games with a unique CCE includes many economically relevant environments beyond those with dominant actions. We illustrate this point with contests, games where players compete for a prize by exerting costly effort. Such models arise in rent seeking, political campaigns, litigation, R\&D races, and sporting competitions. In these settings, a natural concern is that players may
coordinate on low effort and effectively share the prize while saving costs. Uniqueness of the CCE is therefore especially appealing, as it removes the scope for such collusive behavior.

Since contests, as well as many other standard economic models, typically involve a continuum of actions, we begin by noting that the same condition~\eqref{eq_CCE_check} that guarantees uniqueness of $\CCE$ in finite games continues to do so for general games with bounded measurable utilities.

To see this, fix a game $G=(A,u)$ in which each action set $A_i$ is a measurable space and each utility $u_i\colon A\to \R$ is a bounded measurable function. Suppose that condition~\eqref{eq_CCE_check} holds for some $a^*\in A$. Plugging $a_{-i}=a_{-i}^*$ into~\eqref{eq_CCE_check} yields $u_i(a_i,a_{-i}^*)-u_i(a^*)<0$ for every player $i$ and every $a_i\neq a_i^*$. Hence, $a^*$ is a strict Nash equilibrium and, in particular, a pure CCE. To show uniqueness, let $\mu\in \Delta(A)$ be any distribution assigning positive mass to $A\setminus\{a^*\}$. Taking expectations of both sides of~\eqref{eq_CCE_check} with respect to $a\sim \mu$, we obtain $\E_{a\sim\mu}\!\left[\sum_{i\in N}\gamma_i\bigl(u_i(a)-u_i(a_i^*,a_{-i})\bigr)\right]<0$. Therefore, there exists a player $i$ such that $\E_{a\sim\mu}\!\left[u_i(a)-u_i(a_i^*,a_{-i})\right]<0$, which means that~$i$ can increase her expected utility by deviating~to $a_i^*$. Hence, $\mu$ is not a $\CCE$. Boundedness of $u$ in this argument is only used to ensure that expected utilities with respect to $\mu$ are well defined.
\smallskip

First, consider the classical Tullock contest, where each player chooses an  effort level $a_i> 0$ and receives the fraction of the prize equal to her relative effort, $\frac{a_i}{a_1+a_2}$. Assuming that the prize has value $1$ for both contestants and that effort costs are linear, utilities are
$$
u_i(a)=\frac{a_i}{a_1+a_2}-a_i.
$$
The effort profile $a^*=(1/4,1/4)$ is a Nash equilibrium. In fact, it is the unique CCE of this game. To see this, we substitute $u_i$ into~\eqref{eq_CCE_check} with 
$\gamma_i=1$, use that 
$\frac{a_1}{a_1+a_2}+\frac{a_2}{a_1+a_2}=1$ and obtain the following expression for the local potential
\begin{equation}\label{eq_Tulloc_inequality}
   \Phi(a)= \left(\frac{3}{4}-a_1-\frac{1}{1+4a_1}\right)
+\left(\frac{3}{4}-a_2-\frac{1}{1+4a_2}\right).
\end{equation}
Each parenthetical term attains its global maximum of $0$ at $a_i=a_i^*=1/4$, so $\Phi(a)<0$ for all $a\neq a^*$. We conclude that $a^*$ is indeed the unique CCE.

A plausible explanation for why the decentralized (Nash) outcome in the Tullock contest is unaffected by collusion is that the competition is not very intense and exhibits under-dissipation of rents: total equilibrium effort, $1/4+1/4$, is only half the prize value, so the scope for coordinated gains is limited. However, uniqueness of CCE persists even in contests with high rent dissipation. To see this, consider a version of the Tullock contest where $i\in \{1,2\}$ receives the share of the prize equal to 
\begin{equation}\label{eq_Tullloc_modified}
    \frac{(a_i)^r}{(a_1)^r+(a_2)^r},
\end{equation}
where $r> 0$ controls how discriminatory the contest is. The standard Tullock contest corresponds to $r=1$. For large $r$, the contest approaches a perfectly discriminating, winner-takes-all limit; for $r$ close to $0$, the success function becomes flat with weak incentives to exert effort. For $r\le 2$, this contest admits a strict Nash equilibrium $a^*=(r/4,r/4)$, so rent dissipation becomes asymptotically complete as $r$ approaches  $2$. A straightforward modification of the computation for the standard Tullock contest leading to~\eqref{eq_Tulloc_inequality}  shows that, for every $r\in(0,2]$, this Nash equilibrium is in fact the unique CCE.

The uniqueness of the CCE in contests turns out to be a remarkably robust phenomenon. Although the analysis above appears to rely on specific features of the Tullock contest, the conclusion is not tied to its functional form, symmetry, or linearity of costs.

Indeed, consider two players $i\in\{1,2\}$ who participate in a possibly asymmetric contest, exert efforts $a_i>0$, and may differ both in their valuation of the prize and in their effort costs. Let $p_i(a)$ denote the share of the prize awarded to player $i$ as a function of the effort profile $a$. The functions $p_1$ and $p_2$ are called contest success functions, and we assume that they satisfy $p_1(a)+p_2(a)=1$, so the entire prize is always allocated. Player $i$ values the prize at $v_i>0$ and incurs effort cost $c_i(a_i)$. Utilities are therefore given by
\begin{equation}\label{eq_assym_contest}
u_i(a)=v_i\, p_i(a)-c_i(a_i).
\end{equation}
The phenomenon observed in the modified Tullock contest~\eqref{eq_Tullloc_modified}---that the CCE is unique whenever the contest admits a strict Nash equilibrium---extends beyond that example.
\begin{proposition}\label{prop_contest}
    If a contest~\eqref{eq_assym_contest} admits a pure Nash equilibrium $a^*$ that is strict, then $a^*$ is the unique CCE.
\end{proposition}
\begin{proof}
   Suppose $a^*$ is a strict pure Nash equilibrium. We show that the local potential $\Phi$ from \eqref{eq_CCE_check} with $\gamma_i=1/v_i$ is negative for $a\ne a^*$. We get
   $$\Phi(a)=\frac{1}{v_1}\Big(u_1(a_1,a_2)-u_1(a_1^*,a_2)\Big)+ \frac{1}{v_2}\Big(u_2(a_1,a_2)-u_2(a_1,a_2^*)\Big).$$
   Since, $a^*$ is a strict Nash equilibrium, $u_i(a_i,a_{-i}^*)<u_i(a^*)$ for all $a_i\ne a_i^*$. Weight these inequalities by $1/v_i$ and sum to obtain
$$ \frac{1}{v_1}\Big(u_1(a_1,a_2^*)-u_1(a^*)\Big)+ \frac{1}{v_2}\Big(u_2(a_1^*,a_2)-u_2(a^*)\Big)<0.$$
A convenient property of contests is that the left-hand side above coincides identically with $\Phi$. To see this, substitute the expression for $u_i$ from \eqref{eq_assym_contest} and use the identity $p_1+p_2=1$.
We conclude that $\Phi(a)<0$ for $a\ne a^*$, and thus $a^*$ is the only CCE.
\end{proof}

The proposition implies that in contests where players do not benefit from being unpredictable (so a pure Nash equilibrium exists), they also cannot collude to reduce effort. We stress that this result holds under fairly weak assumptions. The proof does not even use monotonicity of the effort costs or positivity of the prize shares. The only property we used is that $p_1+p_2\equiv \mathrm{const}$. One could even allow negative $p_i$, interpreted as each player owning some amount of the prize and transferring part of it to the other player as a function of the effort profile.

Another corollary of Proposition~\ref{prop_contest} is that, whenever a contest admits a strict Nash equilibrium $a^*$, there are no other pure or mixed Nash equilibria. Indeed, $a^*$ must be the only CCE and hence the only Nash equilibrium. This holds even in asymmetric contests and with asymmetric contestants, where multiple equilibria are especially natural to expect. 

The condition of having a unique pure Nash equilibrium can be expressed as an explicit condition on the contest success functions $p_i$. To illustrate this, we restrict attention to the family of symmetric ratio-based contests with linear costs. In such contests, the success function takes the form 
\begin{equation}\label{eq_ratio_contest}
p_1(a)=f\left(\frac{a_1}{a_2}\right)\qquad\text{and}\qquad p_2(a)=f\left(\frac{a_2}{a_1}\right)\qquad  \text{with }\ \  f\left(\frac{a_1}{a_2}\right)+f\left(\frac{a_2}{a_1}\right)=1.    
\end{equation}
For example, this family contains the modified Tullock contest~\eqref{eq_Tullloc_modified} as well as its winner-takes-all limit. Assume the prize value is $1$ for both contestants and effort cost equals effort. Then requiring $a^*=(c,c)$ with $0<c\leq 1/2$ to be a strict Nash equilibrium reduces to $f(a_1/c)-a_1<f(1)-c$ for $a_1\ne c$. By the identity $f(t)+f(1/t)=1$, the values of $f$ on $(0,1]$ determine its values on $[1,+\infty)$; moreover, $f(1)=1/2$. Writing $t=a_1/c$, this becomes $f(t)<1/2-c(1-t)$ for $t\in (0,1)\cup (1,+\infty)$. Rewriting the condition for $t>1$ as a condition on $(0,1)$ yields $1/2+c(1-1/t)<f(t)$.

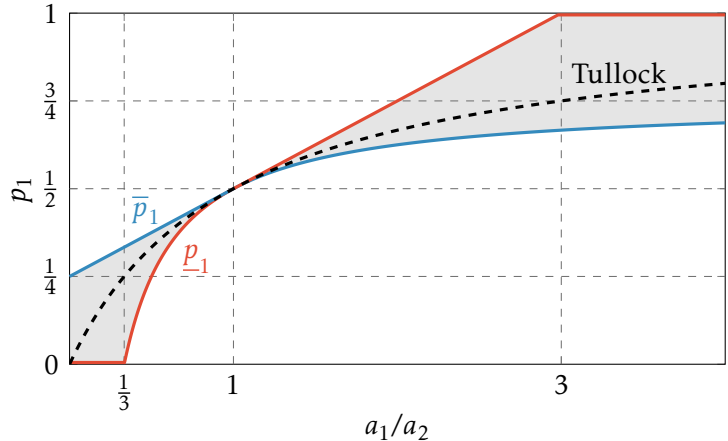
\begin{figure}[h]
\centering

% Set the right endpoint of the x-axis here:
\def\tmax{4} % e.g., 4, 5, 10, ...
\begin{tikzpicture}[scale=0.7]
\begin{axis}[
    width=0.85\linewidth,
    height=0.50\linewidth,
    xmin=0, xmax=\tmax,
    ymin=0, ymax=1,
    domain=0.001:\tmax, 
    samples=500,
    xlabel={$a_1/a_2$},
    ylabel={$p_1$},
    grid=none,
    minor tick num=0,
    xtick={1/3,1,3},
    xticklabels={$\frac{1}{3}$,1,3},
    ytick={0,1/4,1/2,3/4,1},
    yticklabels={$0$, $\frac{1}{4}$,$\frac{1}{2}$,$\frac{3}{4}$, $1$},
]

% Bounds for c=1/4, clipped to [0,1] for plotting:
%   Upper:  min(1, 1/4 + t/4)
%   Lower:  max(0, 3/4 - 1/(4t))

% Define paths (hidden) for shading
\addplot[name path=LB, draw=none] {max(0, 3/4 - 1/(4*x))};
\addplot[name path=UB, draw=none] {min(1, 1/4 + x/4)};

% Shade region between bounds
\addplot[draw=none, fill=black!40, fill opacity=0.25]
    fill between[of=LB and UB];
%\addlegendentry{Region between bounds}

% Plot bounds on top

\addplot[black!55, dashed, thin] coordinates {(1/3,0) (1/3,1)};
\addplot[black!55, dashed, thin] coordinates {(1,0)   (1,1)};
\addplot[black!55, dashed, thin] coordinates {(3,0)   (3,1)};

\addplot[black!55, dashed, thin] coordinates {(0,1/4) (\tmax,1/4)};
\addplot[black!55, dashed, thin] coordinates {(0,1/2) (\tmax,1/2)};
\addplot[black!55, dashed, thin] coordinates {(0,3/4) (\tmax,3/4)};

\addplot[very thick, color=myred, domain=0:1] {max(0.005, 3/4 - 1/(4*x))};
\addplot[very thick, color=myblue, domain=1:\tmax] {max(0, 3/4 - 1/(4*x))};
%\addlegendentry{Lower bound}

\addplot[very thick, color=myblue, domain=0:1] {1/4 + x/4};
\addplot[very thick, color=myred, domain=1:\tmax] {min(0.995, 1/4 + x/4)};
%\addlegendentry{Upper bound}

% Tullock success function f(t)=t/(1+t)
\addplot[very thick, dashed] {x/(1+x)};
%\addlegendentry{Tullock $f(t)=\frac{t}{1+t}$}

\node[text=myblue, anchor=west] at (axis cs:0.32,0.44) {$\overline{p}_1$};
\node[text=myred,  anchor=west] at (axis cs:0.63,0.3) {$\underline{p}_1$};
\node[text=black,  anchor=west] at (axis cs:3.0,0.83) {Tullock};
\end{axis}
\end{tikzpicture}

\caption{Success functions in ratio-based contests for which the symmetric pure Nash equilibrium is the unique pure CCE and the equilibrium effort equals $1/4$. The dashed curve is the Tullock contest. The red curve, $\underline{p}_1$, is the most meritocratic contest. The blue curve, $\overline{p}_1$, is the least discriminatory: each contestant receives a participation bonus of at least $1/4$ of the prize.
\label{fig_contest_range}
}
\end{figure}

Putting this together, a symmetric ratio-based contest~\eqref{eq_ratio_contest} has a strict Nash equilibrium $a^*=(c,c)$ with $c\in (0,1/2]$ (and thus $a^*$ is the unique CCE) if and only if
\begin{equation}\label{eq_condition_on_f}
    \frac{1}{2}+c\left(1-\frac{1}{t}\right)< f(t)< \frac{1}{2}-c\left(1-t\right) \qquad \text{for $t\in (0,1)$}.
\end{equation}
The upper and lower bounds are visualized in Figure~\ref{fig_contest_range} for $c=1/4$, the equilibrium effort in the standard Tullock contest. The bounds themselves correspond to ratio-based contests with success functions $\overline{p}_i$ and $\underline{p}_i$. For $c=1/4$, the success function of the first contestant is given by
$$\overline{p}_1(a)=\begin{cases}
    \frac{1}{4}+\frac{1}{4}\frac{a_1}{a_2}, & \frac{a_1}{a_2}\leq 1\\
    \frac{3}{4}-\frac{1}{4}\frac{a_2}{a_1}, & \frac{a_1}{a_2}\geq 1
\end{cases}\qquad\text{and}\qquad 
\underline{p}_1(a)=\begin{cases}
    0,& \frac{a_1}{a_2}\leq \frac{1}{3}\\
    \frac{3}{4}-\frac{1}{4}\frac{a_2}{a_1}, & \frac{1}{3}\leq \frac{a_1}{a_2}\leq 1\\
     \frac{1}{4}+\frac{1}{4}\frac{a_1}{a_2}, & 1\leq \frac{a_1}{a_2}\leq 3\\
        1,& \frac{a_1}{a_2}\geq 3
\end{cases}.$$
Condition~\eqref{eq_condition_on_f} means that the success function lies between those of $\overline{p}$ and $\underline{p}$. In the contest $\underline{p}$, the low-effort contestant receives the pointwise smallest share of the prize among all ratio-based contests admitting a pure Nash equilibrium. Conversely, in $\overline{p}$ this share is the largest among all such contests: each contestant receives at least $1/4$ of the prize, which can be interpreted as an effort-independent participation bonus. For example, the contest $p=(1-\varepsilon)\overline{p}+\varepsilon \underline{p}$ for small $\varepsilon>0$ preserves the equilibrium effort of the Tullock contest, has a unique CCE, and ensures that a low-effort player receives at least $\frac{1}{4}(1-\varepsilon)$ of the prize pointwise, which can be interpreted as an improvement in fairness.

\medskip

So far, we have focused on games whose unique CCE is {pure}. However, many games of interest---for example, highly competitive contests such as the Tullock contests~\eqref{eq_Tullloc_modified} with $r>2$---do not admit pure Nash equilibria. We now discuss what can be said about CCE uniqueness without assuming purity. As we will see, focusing on pure equilibria is (almost) without loss.

The first observation is that a unique CCE can indeed be mixed. Consider the classical matching pennies game. Its mixed Nash equilibrium is the unique CCE, which can be verified directly by solving the corresponding linear program. More generally, let $G$ be a $2\times 2$ game with a unique Nash equilibrium in which both players randomize. We will refer to such a game as \textbf{a game of matching-pennies type}. Up to relabeling of actions, every game of this type takes the form presented in Table~\ref{tab:counterexample}.
\begin{table}[h]
	\centering
	\begin{tabular}{c|c|c}
		& $a_2$&$b_2$ \\ \hline
		$a_1$&  $(a,e)$& $(b,f)$\\
		$b_1$&  $(c,g)$& $(d,h)$\\
	\end{tabular},\qquad where\qquad $a>c$,\ \ $d>b$,\ \  $f>e$,\ \ $g>h$
	\caption{A game of matching-pennies type has a unique CCE which is mixed.}
	\label{tab:counterexample}
\end{table}
%$$ MATRIX\qquad \text{with }\qquad INEQUALITIES$$

The mixed Nash equilibrium in a game of matching-pennies type is a unique CCE. While this can be shown directly, we give an argument that provides additional insight into the structure of~$\CCE$.

We first observe that $\CCE(G)=\CE(G)$ for all games $G$ where each player has only two actions. Indeed, consider a game $G=(A,u)$ with $n$ players and $|A_i|=2$ for all players $i$. It suffices to show that if $\mu\notin \CE(G)$, then $\mu\notin \CCE(G)$. Since $\mu$ is not a CE, there exists a player $i$ who benefits from deviating after some recommendation; say, when $a_i$ is recommended, she can increase her utility by switching to $b_i$. Because $A_i=\{a_i,b_i\}$, the player can then increase her expected utility by ignoring the recommendation and always playing $b_i$. Hence $\mu$ is not a CCE. Therefore $\CCE(G)\subseteq \CE(G)$, and since always $\CE(G)\subseteq \CCE(G)$, we conclude that the two sets are equal.

Correlated equilibria in $2\times 2$ games were classified by~\cite*{calvo2006}. In particular, he showed that games of matching-pennies type have a unique correlated equilibrium. Thus, for such $G$, the set $\CCE(G)=\CE(G)$ consists of the unique mixed Nash equilibrium. In fact, this $2\times 2$ randomization exhausts all games with a unique CCE that is mixed.
\begin{proposition}\label{prop_unique_CCE_mixed}
If $\CCE(G)$ is unique, then
\begin{itemize}
  \item either it is a pure Nash equilibrium,
  \item or exactly $2$ players mix over $2$ actions each, and the induced $2\times 2$ game is of matching-pennies type.
\end{itemize}
\end{proposition}

By this proposition, any unique CCE is ``almost pure,'' which justifies our focus on pure $\CCE$ in the preceding discussion. Moreover, the $2\times 2$ randomization case is excluded in symmetric environments since no game of matching-pennies type is symmetric.
\begin{corollary}
If $G$ is symmetric and $\CCE(G)$ is unique, then it is a pure Nash equilibrium.
\end{corollary}
Thus, in many environments of interest, the focus on pure Nash equilibria when considering  a unique $\CCE$ is without loss of generality.

\smallskip
The proof of Proposition~\ref{prop_unique_CCE_mixed} is substantially more subtle than the proof of Proposition~\ref{pr:IRCPpure} establishing purity of the unique IRCP. Our argument hinges on an observation about extreme points of $\CCE(G)$, which applies beyond games for which $\CCE(G)$ is a singleton and may be of independent interest.

Since $\CCE(G)$ is defined by a system of linear inequalities, it is convex. A point of a convex set is \textbf{extreme} if it cannot be expressed as a convex combination of other points in the set.
In the finite games we consider, $\CCE(G)$ is a polytope in a finite-dimensional space, and its extreme points are simply its vertices.

Extreme points are relevant for uniqueness since, when $\CCE(G)$ is a singleton, its unique element is a Nash equilibrium and (trivially) extreme, because the polytope degenerates to a single point.\footnote{Another reason for the importance of extreme points is the Bauer maximum principle. It states that any linear or convex objective attains its optimum at an extreme point. Thus non-extreme points can effectively be ignored by a designer. This perspective has been useful in a range of economic theory questions; see the literature review.}
We characterize which Nash equilibria are extreme points of the CCE polytope. A mixed Nash equilibrium $\nu=\nu_1\times\ldots \times \nu_n$ is \textbf{quasi-strict} if all incentive constraints outside its support are slack, that is, $\sum_a \nu(a) u_i(a)> \sum_a \nu(a) u_i(b_i,a_{-i})$ for all $i$ and all $b_i\notin \supp(\nu_i)$. By \cite*{harsanyi1973oddness}, a generic game has finitely many Nash equilibria, and all of them are quasi-strict.
\begin{proposition}
    \label{prop:extremality}
A quasi-strict Nash equilibrium is an extreme point of $\CCE(G)$ if and only if it is pure, or exactly two players randomize and each randomizes over two actions.
\end{proposition}

The proof is in Appendix~\ref{app_extreme}. Here we explain why more extensive randomization is incompatible with extremality.\footnote{A similar tension between extremality and the amount of randomness is observed by \cite*{rudov2025extreme} for Nash equilibria in $\CE$, rather than $\CCE$. Under an additional regularity assumption on the Nash equilibrium, they show that two-player randomization over any two equal numbers of pure actions can be compatible with extremality, whereas three-player randomization cannot.} The idea is to compare the number of binding constraints at a Nash equilibrium in the CCE polytope to the dimension of the space in which the equilibrium lives. By \cite*{winkler1988extreme}, if $k$ linear constraints are imposed on the set of all probability distributions, then any extreme distribution is supported on at most $k+1$ points. Now consider a Nash equilibrium in which $m\leq n$ players randomize. Without loss of generality, these are players $1,\ldots, m$, and player~$i$ mixes over $k_i\geq 2$ actions. The support then has size $k_1\cdot \ldots \cdot k_m$. Quasi-strictness implies that the number of active constraints is at most $k_1+\ldots + k_m$ (players who do not randomize contribute none). Thus a Nash equilibrium can only be extreme if
$$\prod_{i=1}^m k_i\leq 1+\sum_{i=1}^m k_i.$$
Lemma~\ref{lem:combinatorics} in the appendix shows that this inequality can only hold under $2\times 2$ or $2\times 3$ randomization. The latter case does not correspond to extreme equilibria since, since a $2\times 3$ Nash equilibrium must be a convex combination of $2\times 2$ ones \citep{vorob1958equilibrium}.
\smallskip

The classification of unique $\CCE$ in Proposition~\ref{prop_unique_CCE_mixed} follows directly from Proposition~\ref{prop:extremality} and the fact that a unique $\CCE$ must be a quasi-strict Nash equilibrium. This can be deduced from the quasi-strictness of unique $\CE$ established by \cite*{viossat2008having} or proved directly; see Lemma~\ref{lm:CCEstrict} in the appendix.
\begin{proof}[Proof of Proposition~\ref{prop_unique_CCE_mixed}]
    Let $G$ be a game such that $\CCE(G)=\{\nu\}$. By Lemma~\ref{lm:CCEstrict}, $\nu$ is a quasi-strict Nash equilibrium. Since $\CCE(G)$ is a singleton, $\nu$ is an extreme point of $\CCE(G)$. By Proposition~\ref{prop:extremality}, $\nu$ is either pure or involves two players mixing over two actions.

    It remains to show that, in the $2\times 2$ mixing case, the induced $2\times 2$ game $G'$ is of matching-pennies type. Without loss of generality, the mixing players are $i\in\{1,2\}$. Towards contradiction, suppose that $G'$ is not of matching-pennies type. Then, besides the mixed Nash equilibrium $\nu_{12}=\nu_1\times \nu_2$, it has another Nash equilibrium $\eta$. Define
    $\mu=\big((1-\varepsilon)\nu_{12}+\varepsilon \eta\big)\times \nu_{-12}$.
    For $\varepsilon>0$ small enough, $\mu\in \CCE(G)$. Indeed, the incentive constraints for players $1$ and $2$ are satisfied since $\nu_{12}$ and $\eta$ are Nash equilibria of $G'$, and all remaining constraints are slack at $\nu$ by quasi-strictness and therefore remain satisfied for small $\varepsilon$. Since $\mu\neq \nu$, this contradicts the uniqueness of $\CCE$. Hence $G'$ must be of matching-pennies type.
\end{proof}

\medskip

We conclude this section by discussing convexity-based conditions that guarantee a unique $\CCE$. We start with the class of \textbf{strictly socially concave games} introduced by \cite*{even2009convergence}. A game $G=(A,u)$ belongs to this class if each $A_i$ is a compact, convex subset of $\R^d$, the utilitarian welfare $\sum_i u_i(a)$ is strictly concave in $a\in A$, and for every player $i$ the payoff $u_i(a_i,a_{-i})$ is continuous in $a$ and convex in $a_{-i}$.

Any such game is contained in the class of concave games by~\cite*{rosen1965existence}  and thus admits a pure Nash equilibrium $a^*$. As demonstrated by \cite*{hart2015markets}, $a^*$ is the unique $\CCE$. We show that this easily follows from our main result.
\begin{corollary}\label{cor:socially_concave_CCE}
A strictly socially concave game has a unique $\CCE$.
\end{corollary}
Indeed, fix a strictly socially concave game $G=(A,u)$ with a Nash equilibrium $a^*$. To show that the $\CCE$ is unique, we verify that  the local potential $\Phi$ from~\eqref{eq_CCE_check} with $\gamma_i=1$ satisfies $\Phi(a)<0$ for $a\ne a^*$. 
By~\eqref{eq_potential},  $\Phi(a_i,a_{-i}^*)=u_i(a_i,a_{-i}^*)-u_i(a^*)$. Since $a^*$ is a Nash equilibrium, $\Phi(\,\cdot\,,a_{-i}^*)$, as a function of a single variable $a_i$, attains a maximum at $a_i=a_i^*$. By strict concavity, this implies that $a=a^*$ is the unique global maximum of $\Phi$ as a function of the whole profile $a$. Thus $\Phi(a)<\Phi(a^*)=0$ for~$a\ne a^*$. We conclude that $G$ has a unique CCE.

As shown by \cite*{even2009convergence}, the class of strictly socially concave games contains classical economic models such as linear Cournot oligopolies and Tullock contests with strictly convex production or effort costs $c_i$. These correspond to the following utility functions:
$$ \text{Cournot:}\quad  u_i(a)=a_i\Big(\alpha-\beta\sum_j a_j\Big) - c_i(a_i)\qquad \text{Tullock:} \quad u_i(a)=v_i\frac{a_i}{\sum_j a_j} - c_i(a_i).$$
By Corollary~\ref{cor:socially_concave_CCE}, $\CCE$ is unique in these games.

Social concavity is a convenient sufficient condition for a game to be strategically equivalent to an enforcement game, but it is not necessary. For example, a Tullock contest with linear costs, as well as all contests characterized in Figure~\ref{fig_contest_range}, are not strictly socially concave. Nevertheless, these contests have a unique $\CCE$. Moreover, for the Tullock contest with linear costs one can verify~\eqref{eq_CCE_check} with $\gamma_i=1/v_i$, which yields uniqueness for any number of players.
Condition~\eqref{eq_CCE_check} also suggests that for $n\ge 3$ contestants uniqueness of $\CCE$ is more demanding than existence of a pure Nash equilibrium, so an analog of Proposition~\ref{prop_contest} does not hold. Our approach can be used to study the restrictions imposed by uniqueness of $\CCE$ in multi-player contests further, which we leave for future work.

\section{Extensions}\label{sec:extensions}

We model robustness of a strategic interaction outcome to delegation by requiring that the delegation outcome be unique and hence coincide with the unique decentralized Nash outcome. An alternative, and a priori more permissive, way to capture that delegation does not introduce anything new is to assume that the set of delegation outcomes coincides with the convex hull of Nash equilibria. In that case, the intermediary’s role is limited to telling players which Nash equilibrium to play.

For a game $G=(A,u)$, denote by $\conv[\NE(G)]$ the convex hull of Nash equilibria in $\Delta(A)$. Consider a strong intermediary and the associated set of outcomes. Our goal is to characterize games for which
\begin{equation}\label{eq_convex_hull}
    \conv\big[\NE(G)\big]=\IRCP(G).
\end{equation}
If $\IRCP(G)$ is a singleton, then \eqref{eq_convex_hull} holds. The question is whether there are games for which $\IRCP(G)$ is not a singleton but~\eqref{eq_convex_hull} still holds. We show that this never happens unless the game is degenerate, so the two approaches are essentially equivalent.

We fix the set of players $N$ and the set of action profiles $A$, and identify the space of games with $u \in \mathbb{R}^{A \times N}$. We say that a property holds for a \emph{generic game} if the set of games violating it is a nowhere dense subset of $\mathbb{R}^{A \times N}$ with Lebesgue measure zero.
\begin{proposition}\label{prop_hull}
    For a generic game $G$, if $\IRCP(G)$ is not a singleton, then $\conv[\NE(G)]$ is a proper subset of $\IRCP(G)$.
\end{proposition}

The genericity assumption cannot be dropped. For example, if $G=(A,u)$ has $u\equiv 0$, then $\conv[\NE(G)]=\IRCP(G)=\Delta(A)$. A less degenerate example is given in Table~\ref{tab:2x2_degen}.

\begin{table}[ht]
\centering
\begin{tabular}{c|c|c}
		& $a_2$&$b_2$ \\ \hline
		$a_1$& $(1,1)$ & $(1,0)$ \\ 
		$b_1$&  $(0,1)$ & $(1,1)$\\
	\end{tabular}
\caption{A game with two pure Nash equilibria where $\conv[\NE(G)]=\IRCP(G)$}
\label{tab:2x2_degen}
\end{table}

\begin{proof}
By~\cite*{harsanyi1973oddness}, games with a finite set of Nash equilibria, all of which are quasi-strict, are generic.
Toward a contradiction, suppose that $\conv[\NE(G)]=\IRCP(G)$ but $\IRCP(G)$ is not a singleton. Since $\IRCP(G)$ is a polytope, it has at least two distinct extreme points, corresponding to distributions $\nu$ and $\nu'$. Because $\IRCP(G)$ coincides with the convex hull of Nash equilibria, $\nu$ and $\nu'$ are Nash equilibria that are extreme in $\IRCP(G)$ and, by assumption, quasi-strict.

We claim that a quasi-strict Nash equilibrium is an extreme point of $\IRCP(G)$ if and only if it is pure. Indeed, if it is extreme in $\IRCP(G)$, then it is also extreme in $\CCE(G)$ and, by Proposition~\ref{prop:extremality}, it is either pure or involves a $2\times 2$ randomization. The $2\times 2$ mixing case is ruled out as follows. Consider the induced $2\times 2$ game $G'$ and the set of all distributions that give both players the same expected utilities as the mixed Nash equilibrium. Since we impose only two linear constraints, extreme distributions in this set have support of size at most three~\citep*{winkler1988extreme}. Hence the mixed Nash equilibrium is not extreme in this set, not extreme in $\IRCP(G')$, and therefore not extreme in $\IRCP(G)$ by the assumption of quasi-strictness.

It follows that $\nu$ and $\nu'$ are two distinct strict pure Nash equilibria. By~\cite*{gul1993bound}, in a generic game the number of mixed equilibria is at least the number of pure equilibria minus one, so there exists a mixed equilibrium $\nu''$. Since $\nu''$ is mixed, it is not an extreme point of $\IRCP(G)$. However, if $\conv[\NE(G)]=\IRCP(G)$, then $\nu''$ must lie in the convex hull of the extreme points of $\IRCP(G)$, which are pure Nash equilibria. This is impossible: strict pure equilibria cannot be combined to create the indifferences required on the support of a mixed equilibrium.
\end{proof}

Similarly, one can ask about the structure of games such that
$$
\conv[\NE(G)]=\CE(G)\qquad \text{or}\qquad \conv[\NE(G)]=\CCE(G).
$$
The literature provides examples where these identities hold but the respective sets are not singletons.
\cite*{koessler2025correlated} show that $\conv[\NE(G)]=\CCE(G)$ for non-atomic convex potential games, which may have multiple Nash equilibria. Furthermore, $\conv[\NE(G)]=\CE(G)$ for potential games with a finite number of players \citep*{neyman1997correlated, ui2008correlated, cao2025correlated} and auctions \citep*{feldman2016correlated}, which may also have multiple Nash equilibria.  Whether genericity of $G$ is sufficient to rule out such examples---and to obtain that $\CE(G)$ or $\CCE(G)$ is generically larger than the convex hull of Nash equilibria unless it is a singleton---is, in our opinion, a very interesting but likely difficult question. Addressing it would require a deeper understanding of the structure of Nash equilibria, $\CE$, and $\CCE$.

\pagebreak

\bibliography{main}

\appendix
\section{Proof of \texorpdfstring{Proposition~\ref{prop_continuum_IRCP}}{Proposition 2}}\label{app_continuum_IRCP}
The first bullet point of Proposition~\ref{prop_continuum_IRCP} states that games affinely equivalent to enforcement games have unique IRCPs which are pure. For enforcement games themselves,  Lemma~\ref{lm:enf_unique} applies beyond the finite action setting to games in which players have arbitrary measurable action sets and bounded measurable utilities. Indeed, the proof of that lemma works verbatim for this setting. This unique IRCP must be the pure action profile $a^*$ given in the definition of enforcement games, which is necessarily an IRCP. Since positive affine transformations of utilities preserve the defining inequalities of IRCP, the same holds for games affinely equivalent to enforcement games.

We now prove the second bullet point of Proposition~\ref{prop_continuum_IRCP}, which states that the converse holds at this level of generality, provided that utilities are symmetric. 
\begin{proof}[Proof of Proposition~\ref{prop_continuum_IRCP}]
Let $G=(A,u)$ be a symmetric game with measurable action sets and bounded, measurable utilities and maximin levels $\underline{u}_i$. Suppose that $a^*$ is the unique IRCP of $G$. Note that the maximin levels are well-defined by the assumption that each player $i$ has a maximin strategy, which we denote by $\nu_i$. 
Consider the affinely equivalent game $G'=(A,u')$ with~$u'$ given by
\[
    u_i'(a)=u_i(a)-u_i(a^*).
\] In $G'$, each player receives a utility of zero at the unique IRCP $a^*$.  We claim that the maximin levels of $G'$ are zero. Indeed, symmetry implies that the maximin levels of all players coincide and they cannot be less than zero, or a profile $\mu$ that mixed $a^*$ with a sufficiently small probability on any other profile would be a distinct IRCP. Since the product of maximin strategies $\nu_1\times \cdots\times \nu_n$ is an $\IRCP$ we conclude that, by uniqueness, each $\nu_i$ is a point mass at $a_i^*$. It follows that playing $a_i^*$ unilaterally guarantees $i$ a non-negative utility.

All that remains is to show that the welfare at any action profile other than $a^*$ is negative. Indeed, consider $b=(b_1,\dots,b_n)\neq a^*$ and let $\mu$ be the uniform distribution over all profiles $(b_{\pi(1)},\dots, b_{\pi(n)})$ constructed by permuting the components of $b$. %$
By symmetry $\mu$ yields each player identical utility, which is negative, since $\mu$ is not the IRCP $a^*$. However, the welfare at $\mu$ is equal to the welfare at $b$, so $b$ also yields a negative welfare, and we conclude that $G'$ is an enforcement game.
\end{proof}

\section{Proof of \texorpdfstring{Theorem~\ref{th:CCE}}{Theorem 2}}\label{app_CCE}
The proof of Theorem~\ref{th:CCE} relies on the following lemma, which establishes that a unique CCE is a quasi-strict Nash equilibrium.\footnote{A Nash equilibrium is quasi-strict if any player who deviates to an action outside of the support of their mixed-strategy receives a strictly lower expected utility.}
\begin{lemma}\label{lm:CCEstrict}
        A unique CCE is a quasi-strict Nash equilibrium.
\end{lemma}
Lemma~\ref{lm:CCEstrict} follows from Lemma~3 of \cite*{viossat2008having}, which states that any unique correlated equilibrium---which includes any unique CCE---is quasi-strict. For completeness, we provide a short proof of Lemma~\ref{lm:CCEstrict} via a duality argument. 
We will use the following strict complementary slackness result for zero-sum games due to \cite*{shapley1949solutions} and \cite*{arrow1953admissible}; see Theorem 3.11 in a book by \cite*{karlin2003mathematical}.
\begin{theorem}\label{th_Karlin}
Consider a zero-sum game with finite sets of pure strategies $S_1$ and $S_2$ and payoffs $\pi_1=-\pi_2$. Let $V$ be its value and $\Sigma_i^*\subset \Delta(S_i)$ be the set of maximin strategies of player~$i$. Then there exists $\sigma_2\in \Sigma_2^*$ such that $\sigma_2(s_2)>0$ if and only if $\sum_{s_1} \pi_1(s_1,s_2)\sigma_1(s_1)=V
    $ for all $\sigma_1\in \Sigma^*_1$.
\end{theorem}
That is, a pure strategy $s_2$ is played by the second player with positive probability at some maximin strategy if this action belongs to their best response to every maximin strategy of player~1. Without the game-theoretic framing, this result can be seen as an expression of strict complementary slackness in linear programming. The following corollary will be particularly useful for proving Lemma~\ref{lm:CCEstrict}.
\begin{corollary}\label{cor:Karlin}
Consider a zero-sum game $G$ with finite sets of pure strategies $S_1$ and $S_2$ and payoffs $\pi_1=-\pi_2$. Let $V$ be the value of $G$ and suppose that player~1 has a unique maximin strategy $\sigma_1$. Then there exists a maximin strategy of the second player $\sigma_2$ such that $\supp(\sigma_2)=\{s_2\colon \pi_1(\sigma_1,s_2)=V\}$ and $\pi_1(s_1,\sigma_2)<V$ for $s_1\notin\supp(\sigma_1)$.
\end{corollary} 

\begin{proof}[Proof of Lemma~\ref{lm:CCEstrict}]
Let $\mu$ be the unique CCE of $G=(A,u)$, whose set of players is $N=\{1,\dots,n\}$. Consider the following auxiliary zero-sum game between Maximizer, whose set of pure strategies is $A$, and Minimizer who must choose a player $i$ as well as an action of that player $a_i$. For a profile $b$ chosen by Maximizer and a pair $(i,a_i)$ chosen by Minimizer, Maximizer's payoff is 
\[
    u_i(b)-u_i(a_i,b_{-i}).
\] Since $\mu$ is the unique CCE of $G$, it is also the unique maximin strategy of Maximizer. Moreover, $\mu$ must be a Nash equilibrium. Thus, by choosing any player $i$ and action $a_i$ played with positive probability under $\mu$, Minimizer achieves a payoff of zero against $\mu$. It follows that for every mixed-strategy $\nu\in \Delta(A)$, 
\begin{equation}\label{eq:Minimizer_wins}
   \min_{(i,a_i)} \  \sum_{b}\nu(b) \left(u_i(b)-u_i(a_i,b_{-i})\right) \leq 0,
\end{equation} with equality if and only if $\nu=\mu$.

By Corollary~\ref{cor:Karlin}, Minimizer has a mixed-strategy $\tau$ whose support is all the pairs $(i,a_i)$ that lead to zero payoff as a response to $\mu$ and such that Maximizer receives a payoff of zero if she chooses a pure strategy in the support of $\mu$ and a strictly negative payoff under any other pure strategy. We can express $\tau$ as a distribution $\eta\in \Delta(N)$ and a collection of mixed-strategies $\sigma_i\in \Delta(A_i)$, $i\in N$, such that $\tau(i,a_i)=\eta(i)\sigma_i(a_i)$. For any $a_i$ played with positive probability under $\mu$, the pair $(i,a_i)$ is a best response to $\mu$, and thus $\eta(i)\sigma_i(a_i)>0$. In particular, $\eta(i)>0$ for all $i$.

We will now show that $\sigma = \sigma_1\times\cdots\times\sigma_n$ defines a quasi-strict Nash equilibrium. Indeed, let $b_i \in A_i$ and suppose Maximizer plays $\nu = \sigma_{-i}\times \delta_{b_i}$ and Minimizer chooses $\tau$. 

Each $j\neq i$ contributes zero to Maximizer's resulting payoff, since for all $b_{-j}$,
\[
   \eta(j)\left(\sum_{b_j\in A_j}\sigma_j(b_j)u_j(b_j,b_{-j})-\sum_{a_j\in A_j}\sigma_j(a_j)u_j(a_j,b_{-j})\right)=0. 
\] By \eqref{eq:Minimizer_wins}, the contribution from $i$ is
\begin{equation*}
    \eta(i)\sum_{a}\left(u_i(b_i,a_{-i})-u_i(a) \right)\sigma(a)\leq 0,
\end{equation*}
   with equality if and only if $b_i$ is played with positive probability under $\mu$. Since $\eta(i)$ is positive, this is precisely the condition for $\sigma$ to be a quasi-strict Nash equilibrium. Since $\mu$ is the unique Nash equilibrium of $G$, we have $\mu=\sigma$, so $\mu$ is a quasi-strict Nash equilibrium.
\end{proof}
With Lemma~\ref{lm:CCEstrict} in hand, we will now show that Theorem~\ref{th:CCE} follows from Theorem~\ref{th:IRCP}.

\begin{proof}[Proof of Theorem~\ref{th:CCE}]
	One direction is immediate. If $G=(A,u)$ is strategically equivalent to an enforcement game, then it has a unique $\CCE$ since enforcement games have unique $\IRCP$ and thus $\CCE$,  and the set of $\CCE$ is preserved under strategic equivalence.

	We now prove the opposite direction. In particular, we will show that if the action profile $a^*$ is a unique $\CCE$ of $G=(A,u)$ then it is a unique $\IRCP$ in the game $G'=(A,u')$ with utilities given by
    \[
    u_i'(a)=u_i(a)-u_i(a_i^*,a_{-i}).
    \]
    Since $a^*$ is a unique CCE in $G'$, by Lemma~\ref{lm:CCEstrict}, it is a strict Nash equilibrium. Each player's equilibrium utility is $0$ in $G'$. Thus, player $i$'s maximin utility is at most $0$. On the other hand, $a_{i}^*$ guarantees $i$ a utility of $0$, so each player's maximin utility in $G'$ is $0$. Let $m$ be the minimum amount that any player $i$ can lose from deviating from equilibrium to another pure action. 
    
   Let $\mu$ in $\IRCP(G')$, so that the expected utility of player $i$ under $\mu$ is at least $0$ for all $i$. Let $M$ be the maximum amount that any player $i$ can gain from deviating from $\mu$ to another pure action, and note that $M\ge 0$, since otherwise $\mu$ would define a CCE that is distinct from $a^*$. 
   
   Let $\alpha \in (0,\frac{m}{m+M})$ and consider the correlated profile $\nu=(1-\alpha)\delta_{a^*}+\alpha\mu$. Each player's expected utility under $\nu$ is at least $0$, since this holds for both $a^*$ and $\mu$. If some player~$i$ deviates from $\nu$ to some pure action $a_i$, they incur a penalty of at least $(1-\alpha)m$ by deviating from $a^*$ and gain at most $\alpha\cdot M$ from deviating from $\mu$. This is a net loss, meaning $\nu$ is a CCE. Since the only CCE is at $a^*$, we conclude that $\mu=\delta_{a^*}$ and that this is the unique IRCP of $G'$. 
   
   Finally, by Theorem~\ref{th:IRCP}, $G'$ is affinely equivalent to an enforcement game. It follows that $G$ is strategically equivalent to this enforcement game.
\end{proof}

\section{Proof of Proposition~\ref{prop:extremality}}\label{app_extreme}
We will need the following lemma. It will help us reduce the question to considering $2\times 2$ or $2\times 3$ randomization
\begin{lemma}
\label{lem:combinatorics}
Let $k_1, \dots, k_m$ be integers with $k_i \geq 2$ for all $i=1, \dots, m$. If
$$\prod_{i=1}^m k_i\leq 1+\sum_{i=1}^m k_i,$$
then $m\leq 2$. Moreover, if $m=2$, then $\{k_1, k_2\}$ is either $\{2,2\}$ or $\{2,3\}$.
\end{lemma}
\begin{proof}
Divide the inequality by $\prod_{i=1}^m k_i>0$ to obtain
\[ 1 \leq \frac{1}{\prod_{i=1}^m k_i} + \sum_{i=1}^m \frac{1}{\prod_{j \neq i} k_j}. \]
Since $k_i\ge 2$ for all $i$, the right-hand side is at most $\frac{1}{2^m}+\frac{m}{2^{m-1}}=\frac{1+2m}{2^m}$, which is strictly less than $1$ for all $m\ge 3$. Indeed, this expression equals $7/8$ at $m=3$ and decreases thereafter. Thus $m\le 2$.
For $m=2$, the inequality becomes $k_1k_2\le 1+k_1+k_2$, i.e.\ $(k_1-1)(k_2-1)\le 2$. Since $k_i\ge 2$, each factor is at least $1$, so the only possibilities are $\{2,2\}$ and $\{2,3\}$.
\end{proof}

\begin{proof}[Proof of Proposition~\ref{prop:extremality}]
Let $\nu=\nu_1\times \cdots \times \nu_n$ be a quasi-strict Nash equilibrium of $G=(A,u)$. Write $k_i=|\supp(\nu_i)|$, and let $m$ be the number of players with $k_i\ge 2$. Without loss of generality, assume these are players $i=1,\ldots,m$. We prove each direction in turn.

($\Leftarrow$) If $\nu$ is pure ($m=0$), then it is a vertex of the simplex $\Delta(A)$ and hence extreme in $\Delta(A)$. Therefore it is also extreme in the subset $\CCE(G)\subseteq \Delta(A)$.

Now suppose $\nu$ involves a $2\times 2$ randomization, i.e., $m=2$ and $k_1=k_2=2$. Consider the induced $2$-player $2\times 2$ game $G'$ obtained by fixing players $3,\ldots,n$ at their equilibrium actions and restricting players $1$ and $2$ to $\supp(\nu_1)$ and $\supp(\nu_2)$. The equilibrium $\nu$ induces a mixed Nash equilibrium of $G'$. By \cite*{calvo2006}, in a $2\times 2$ game a mixed Nash equilibrium is an extreme point of the correlated equilibrium set. Moreover, in this $2\times 2$ setting $\CE(G')=\CCE(G')$ (see the discussion in Section~\ref{sec:CCE}), so $\nu$ is extreme in $\CCE(G')$.

To conclude, suppose for contradiction that $\nu$ is not extreme in $\CCE(G)$, so $\nu=\alpha \mu'+(1-\alpha)\mu''$ for some $\alpha\in(0,1)$ and distinct $\mu',\mu''\in \CCE(G)$. Since $\nu(a)=0$ for every $a\notin \supp(\nu)$, we must have $\mu'(a)=\mu''(a)=0$ for all $a\notin \supp(\nu)$ as well. Thus $\mu'$ and $\mu''$ can be identified with elements of $\CCE(G')$, and they represent a non-trivial convex decomposition of $\nu$ in $\CCE(G')$, contradicting extremality of $\nu$ in $\CCE(G')$. Hence $\nu$ is extreme in $\CCE(G)$.

\smallskip

($\Rightarrow$) Assume that $\nu$ is an extreme point of $\CCE(G)$. Let $k$ denote the number of linearly independent incentive constraints that are active at $\nu$. By \cite*{winkler1988extreme}, any extreme point of the set of probability distributions satisfying these constraints is supported on at most $k+1$ action profiles. Since only players $1,\ldots,m$ randomize, we have $|\supp(\nu)|=\prod_{i=1}^m k_i$.

By quasi-strictness, for each mixing player $i\le m$ the incentive constraints corresponding to deviations $a_i'\notin\supp(\nu_i)$ are slack at $\nu$. Hence each such player $i$ contributes at most $k_i$ active constraints. For players $i>m$ who play pure, the only constraint with $a_i'$ in the support reduces to $0\ge 0$ and therefore does not contribute to the count of independent active constraints. It follows that
$$\prod_{i=1}^m k_i\leq 1+\sum_{i=1}^m k_i.$$
By Lemma~\ref{lem:combinatorics}, this implies $m\le 2$, and if $m=2$ then $\{k_1,k_2\}$ is either $\{2,2\}$ or $\{2,3\}$.

We next rule out the cases of a single mixing player and of two mixing players, one mixing over~$2$ actions and the other over~$3$.

If $m=1$, then player $1$ is indifferent among all actions in $\supp(\nu_1)$. Choose two distinct mixed strategies $\nu_1',\nu_1''$ with support contained in $\supp(\nu_1)$ such that $\nu_1=\alpha \nu_1'+(1-\alpha)\nu_1''$ for some $\alpha\in(0,1)$. Let
$\mu'=\big((1-\varepsilon)\nu_{1}+\varepsilon \nu_1'\big)\times \nu_{-1}$ and
$\mu''=\big((1-\varepsilon)\nu_{1}+\varepsilon \nu_1''\big)\times \nu_{-1}$.
Since, by quasi-strictness, all constraints involving actions outside $\supp(\nu)$ are slack, we have $\mu',\mu''\in \CCE(G)$ for $\varepsilon>0$ small enough. Then $\nu=\alpha \mu'+(1-\alpha)\mu''$ is a non-trivial convex decomposition, contradicting extremality. Thus $m\neq 1$.

Finally, suppose $m=2$ and $\{k_1,k_2\}=\{2,3\}$. Consider the induced $2$-player subgame obtained by fixing players $3,\ldots,n$ at their equilibrium actions. If in a Nash equilibrium of a two-player game the players mix over different numbers of pure strategies, then this equilibrium lies in the convex hull of equilibria in which they mix over the same number of strategies \citep*{vorob1958equilibrium}. Thus any $2\times 3$ Nash equilibrium in a $2$-player game is a convex combination of $2\times 2$ equilibria. Hence there are two distinct Nash equilibria $\eta',\eta''$ such that $\nu_{12}=\nu_1\times \nu_2$ can be written as $\alpha \eta'+(1-\alpha)\eta''$ for some $\alpha\in(0,1)$. Define
$\mu'=\big((1-\varepsilon)\nu_{12}+\varepsilon \eta'\big)\times \nu_{-12}$ and
$\mu''=\big((1-\varepsilon)\nu_{12}+\varepsilon \eta''\big)\times \nu_{-12}$.
Since all incentive constraints that are slack at $\nu$ remain satisfied for $\varepsilon>0$ small enough, we have $\mu',\mu''\in \CCE(G)$ for sufficiently small $\varepsilon$. Moreover, $\nu=\alpha \mu'+(1-\alpha)\mu''$, again contradicting extremality. Hence the $2\times 3$ case is impossible.

Therefore, if a quasi-strict Nash equilibrium is extreme in $\CCE(G)$, it must be either pure or involve exactly two players mixing over two actions each.

\end{proof}

\section{Connection between IRCP and Guaranteed Utility Equilibria of \cite*{csoka2024guaranteed}}\label{app_GUE}

\cite*{csoka2024guaranteed} introduce the notion of a \emph{guaranteed utility equilibrium} (GUE), which significantly strengthens that of a pure Nash equilibrium.
In a game $G=(A,u)$, a pure action profile $a^*\in A$ is a GUE if it is Pareto optimal among pure profiles and  each player can secure her payoff at $a^*$ unilaterally, that is, one has $u_i(a_i^*,a_{-i})\ge u_i(a^*)$ for every $i$ and $a_{-i}$.

Enforcement games (Definition~\ref{def:enforcement}) constitute a class of games admitting a GUE. Indeed, the self-enforcing profile $a^*$ maximizes utilitarian welfare, hence it is Pareto optimal. Moreover, self-enforcement implies that player $i$ guarantees $u_i(a^*)$ by playing $a_i^*$ regardless of what others do. Since both properties are preserved under positive affine transformations, any game affinely equivalent to an enforcement game admits a GUE. Combining this observation with Theorem~\ref{th:IRCP}, we obtain that whenever a game has a unique $\IRCP$, this unique $\IRCP$ is a GUE.

The notion of GUE nevertheless differs from unique $\IRCP$ in two ways. First, there may be several distinct GUE action profiles that induce the same payoff vector; thus a game admitting a GUE may have multiple $\IRCP$. This multiplicity is easily avoided by imposing a strictness requirement: no other pure profile $a\neq a^*$ yields the same payoff vector $u(a)=u(a^*)$.

Second, and more importantly, the Pareto requirement in the definition of GUE compares $a^*$ only to other pure profiles. It does not rule out the possibility that a lottery $\mu\in\Delta(A)$ Pareto-improves upon $a^*$. In that case, the $\IRCP$ need not be a singleton even if the game admits a strict GUE.
\begin{table}[h]
\centering
\begin{tabular}{c|ccc}
 & $a_2$ & $b_2$ & $c_2$\\\hline
$a_1$ & $(0,0)$ & $(0,-1)$ & $(0,-1)$\\
$b_1$ & $(-1,0)$ & $(2,-1)$ & $(-1,-1)$\\
$c_1$ & $(-1,0)$ & $(-1,-1)$ & $(-1,2)$
\end{tabular}
\caption{A game with a unique GUE profile but non-unique $\IRCP$ and non-unique $\CCE$.}
\label{tab:gue_not_unique_ircp_cce_3x3}
\end{table}

To see this, consider the two-player game in Table~\ref{tab:gue_not_unique_ircp_cce_3x3}. The profile $(a_1,a_2)$ with utilities $(0,0)$ is the only GUE. In particular, it is Pareto optimal among pure profiles because every other profile gives at least one player payoff $-1$. 
However, lottery $\mu=\tfrac12(b_1,b_2)+\tfrac12(c_1,c_2)$ yields expected payoff $(1/2,1/2)$ and thus dominates $(a_1,a_2)$. Consequently, $\IRCP$ contains both $(a_1,a_2)$ and $\mu$
and is not a singleton. Moreover, $\mu$
is also a CCE, and so $\CCE$ is not unique either.

This example suggests strengthening the Pareto requirement and adding strictness in the definition of GUE. 
We call $a^*$ a \emph{strict fractional GUE} if it is Pareto optimal among lotteries over~$A$, no other lottery leads to the same utility profile, and each player can secure her payoff at $a^*$ unilaterally.

\begin{lemma}
A game $G$ is affinely equivalent to an enforcement game if and only if $G$ admits a strict fractional GUE.
\end{lemma}

\begin{proof}
By Theorem~1 of \cite*{arrow1953admissible}, $a^*$ is fractionally Pareto optimal if and only if there exist weights $\gamma_i>0$ such that $a^*$ maximizes $\sum_i \gamma_i u_i(a)$ over $a\in A$. In an enforcement game the self-enforcing profile $a^*$ is, by definition, the unique maximizer of $\sum_i u_i(a)$, so it is fractionally Pareto optimal and strict. Therefore it is a strict fractional GUE.

Conversely, let $a^*$ be a strict fractional GUE in $G$. The theorem by \cite*{arrow1953admissible} provides weights $\gamma_i>0$ such that $a^*$ maximizes $\sum_i \gamma_i u_i(a)$. Moreover, we may take these weights so that the maximizer is unique.\footnote{Strictness rules out other $\mu\in\Delta(A)$ with $\sum_a \mu(a)u(a)=u(a^*)$, so $u(a^*)$ is a vertex of $\conv\{u(a):a\in A\}$ and hence exposed. Let $\beta$ expose it. If $\gamma$ are \cite*{arrow1953admissible} weights, then for $\varepsilon>0$ small the perturbation $(1-\varepsilon)\gamma+\varepsilon\beta$ stays positive and makes $a^*$ the unique maximizer.} 
Consider the affinely equivalent game $H=(A,v)$ defined by $v_i(a)=\gamma_i\bigl(u_i(a)-u_i(a^*)\bigr)$. Then $a^*$ uniquely maximizes $\sum_i v_i(a)$ and $v(a^*)=0$. Moreover, each player can guarantee payoff $0$ in $H$ by playing $a_i^*$. Thus $H$ is an enforcement game.
\end{proof}

Comparing this lemma with Theorem~\ref{th:IRCP}, we conclude that a game has a unique $\IRCP$ if and only if it admits a strict fractional GUE.

\end{document}